\documentclass[12pt]{amsart}
\usepackage{amsmath, amsthm, amsfonts, amssymb, epsfig, setspace,color}
\newtheorem{thm}{Theorem}[section]
\newtheorem{lem}[thm]{Lemma}
\theoremstyle{definition}
\newtheorem{defn}[thm]{Definition}

\newtheorem{example}{Example}
\newtheorem{cor}[thm]{Corollary}
\newtheorem{prop}[thm]{Proposition}
\newtheorem{remark}[thm]{Remark}

\DeclareMathOperator{\Span}{sp}
\DeclareMathOperator{\Supp}{Supp}

\newcommand{\MC}{\mathcal}
\newcommand{\MB}{\mathbb}
\newcommand{\MBF}{\mathbf}
\newcommand{\TI}{\textit}
\newcommand{\bF}{\Bbb F}
\newcommand{\bZ}{\Bbb Z}

\begin{document}
%\doublespacing
\title{Higher genus universally decodable matrices  (UDMG)}
\author{Steve Limburg, David Grant, Mahesh K. Varanasi}
\address{Department of Mathematics, University of Colorado at Boulder,
Boulder, Colorado 80309-0395 USA}\email{limburg@colorado.edu}
\address{Department of Mathematics, University of Colorado at Boulder,
Boulder, Colorado 80309-0395 USA}\email{grant@colorado.edu}
\address{Department of Electrical and Computer Engineering, University of Colorado at Boulder,
Boulder, Colorado 80309-0425 USA} \email{varanasi@colorado.edu}
\date{\today}

\thanks{The first author was partially supported by Department of Education GAANN grant P200A060220.}

\begin{abstract}

%[why are we using $g$ instead of $g$?]

We introduce the notion of Universally Decodable Matrices of Genus $g$ (UDMG), 
which for $g=0$ reduces to the notion of Universally Decodable Matrices (UDM) introduced in \cite{VG}. A UDMG is a
set of $L$ matrices over a finite field $\MB F_q$, each with $K$ rows, and a linear independence condition satisfied by 
collections of  $K+g$ columns formed from the initial segments of the matrices. We consider the mathematical structure
of UDMGs and their relation to linear vector codes. 
We then give a construction of UDMG based on 
curves of genus $g$ over $\MB F_q$, which is a natural generalization of the UDM constructed in \cite{VG} from $\MB P^1$. 
We provide upper (and constructable lower) bounds for $L$ in terms of $K$, $q$, $g$, and the number of columns of the matrices.% these improve the bounds given in \cite{VG} in the case $g=0$. 
We will show there is a fundamental trade off (Theorem \ref{bounddelta}) between $L$ and $g$, akin to the Singleton bound for the minimal
Hamming distance of linear vector codes.
\end{abstract}
\maketitle
\section*{Introduction}
Universally Decodable Matrices (UDM) over finite fields were introduced by Tavildar and Viswanath in \cite{Tal} to build examples of  {\it approximately universal codes} (defined below), which were designed to solve an important problem in coding over parallel channels in slow-fading wireless communications systems.
Recently Vontobel and Ganesan gave a general construction for UDMs in \cite{VG}.  (See also \cite{VG1}.)

In this paper we introduce a natural and useful generalization of UDMs we call {\it Universally Decodable Matrices of Genus $g$} (UDMGs).
We then generalize the construction of UDMs given in \cite{VG} and find bounds for the size of a UDMG that apply in a more general setting than that considered in \cite{Tal} and \cite{VG}.

Despite (or perhaps because of) their utilitarian origin, these sets of matrices can be studied as an abstract mathematical structure in their own right, and as such have a rich and beautiful theory, 
including relations to traditional linear vector codes --- which in some sense they generalize. 
Before we detail this structure, let us describe in more detail the communications problem
which inspired their consideration and to which they provide a solution.

\subsection*{Communication Motivation for UDMG}
First let us review the terminology we need from communications theory. 
%We can think of a radio wave to be transmitted over a wireless channel
%as a complex number, the real part giving the amplitude and the imaginary part the frequency of the wave (the {\it channel} itself is defined probablistically.)
Digital communication over a wireless channel takes place via the transmission of complex numbers whose magnitude and 
argument determine the amplitude and phase of the radio-frequency wave over which they are transmitted (the channel itself is randomly time-varying and is defined probabilistically).
 %We will assume that the channel is linear, 
 The radio wave is received by an antenna %and is down-converted to baseband 
and sampled at the symbol rate,
 so if $x\in \Bbb C$ is the transmitted 
information-bearing complex-valued symbol, the corresponding discrete-time complex %baseband 
received signal will
be $y=hx+n$, where $h$ and $n$ are realizations of complex random variables, respectively called the {\it fading coefficient} and the {\it noise} 
of the channel. We assume that the noise is an additive complex Gaussian random variable
of mean 0 and variance 1. If the realization $h$ is constant over all $T$ timeslots that we will employ the channel, we say the channel is {\it slow-fading}.
A set of $L$ channels is called a {\it parallel} channel, and its elements are called its {\it subchannels}.
 An important example of a parallel channel is one that results in wide-band communication through the use of a technique called orthogonal frequency division multiplexing (OFDM) \cite{tse}.

Therefore given a set $W$ of messages (information), we can transmit it over $L$-parallel subchannels for $T$ 
timeslots via an injection 
$i:W\rightarrow \text{Mat}_{L\times T}(\Bbb C)$.
%(which we can extend entry-by-entry to a map $W^T\rightarrow \text{Mat}_{L\times T}(\Bbb C)$ over $T$ timeslots). 
There is a great deal of application-specific engineering that goes into constructing $i$, and it is useful to write it as the composition of an {\it encoding} map 
$\kappa$ from $W$ into a set $C$ of {\it codewords}, and a map $\mu: C\rightarrow \text{Mat}_{L\times T}(\Bbb C)$ called {\it modulation}.
% (both of which can be extended entry-by-entry to maps $\kappa:W^T\ra C^T$ and $\mu:C^T\ra \text{Mat}_{L\times T}(\Bbb C)$). 
 We will call the quadruple $(W,\kappa,C,\mu)$ a
 {\it coding scheme} (or just a {\it code}). The {\it rate} of the code is $\log_2|W|/T$. The {\it power} of the code is 
${1\over{T |C|}}\displaystyle \sum_{x\in C}||\mu(x)||
 ^2,$ where $||~ \cdot ~||$ denotes the Frobenius norm, which by our normalizing choice of the noise
is the same as the {\it signal-to-noise ratio} (SNR), which we denote as $SNR(C,\mu)$. 
%Below we will describe a modulation map $\mu$ of most interest to us.
%[I'm assuming rate and SNR are defined per timeslot, and so are independent of the choice of $T$: is this right?]

Recently \cite{Tal} gave a definition of what it means for a sequence of codes for a slow-fading parallel channel to be ``approximately universal"  (for the experts, this was
meant to capture the notion of what it means for the sequence of codes to optimally trade off diversity and multiplexing gain, no matter the choice of the distribution of the fading coefficients). 
So as to not bring us too far afield, we will use an operational definition of approximately universal given in Theorem $5.1$ of \cite{Tal}:

Suppose we have a slow-fading parallel channel with $L$ subchannels.
For each natural number $n$, suppose we have a coding scheme $(W_n,\kappa_n,C_n,\mu_n)$ for our channel,
employed for $T$ timeslots, of rate $R_n$, and signal-to-noise ratio $SNR_n$, such that $SNR_n$ tends to $\infty$ as $n$ does.  Then we say the
sequence is {\it approximately universal} if for every pair of distinct $T$-tuples of codewords $v,w\in C_n^{T}$,
$$\prod_{1\leq i\leq L} \parallel d_i \parallel^2 \geq \frac{1}{2^{R_n+o(\log(SNR_n))}},$$ 
where $d_i$ is the $i^{th}$-row of the $L\times T$  matrix $(\mu_n(v)-\mu_n(w))/\sqrt{SNR_n}$. %, and $\parallel \cdot \parallel$ denotes the euclidean norm.

Given a coding scheme $(W,\kappa,C,\mu)$ for one timeslot, for any $T$,
we can extend $\kappa$ and $\mu$ entry-by-entry to functions $\kappa^T$ and
$\mu^T$ of the vectors $W^T$ and $C^T$, to get the {\it $T$-iterated} coding scheme $(W^T,\kappa^T,C^T,\mu^T)$ for $T$
timeslots.
Using the arithmetic-geometric mean inequality as in the proof of the following Lemma, it is not hard to see that if a sequence 
of coding schemes for one timeslot is approximately universal, then for any $T$, the corresponding sequence of $T$-iterated coding
schemes is
approximately-universal. So for the purpose of building examples of sequences of approximately universal coding schemes,
it suffices to build examples for one timeslot. So we will assume henceforth that $T=1$. We also need the following simplification:

\begin{lem}\label{complex} 
For each natural number $n$, suppose we have a coding scheme $(W_n,C_n,\kappa_n,\mu_n)$ for our parallel channel with $L$-subchannels, 
of rate $R_n$, and signal-to-noise ratio $SNR_n$, such that $SNR_n$ tends to $\infty$ as $n$ does, and that $\mu_n$ is real-valued.
Suppose that for every pair of distinct codewords $v,w\in C_n$,
$$\prod_{1\leq i\leq L} d_i ^2 \geq \frac{1}{2^{2R_n+o(\log(SNR_n))}},$$ 
where $d_i$ is the $i^{th}$-entry of the vector $(\mu_n(v)-\mu_n(w))/\sqrt{SNR_n}$.
Then the {\it complexified} coding scheme $$(W_n\times W_n,\kappa_n\times\kappa_n,C_n\times C_n,\mu_n\times 0+0\times i\mu_n)$$
%for any coding scheme $(W,C,\kappa,\mu)$, the modulation map $\mu: C\rightarrow \Bbb C: can be written as $\mu=\mu_1+i\mu_2$ where $\mu_1,\mu_2: C\rightarrow \Bbb R$ %are respectively called $\it amplitude modulation$ and $\it frequency modulation$. Conversely, given any coding schemes
is approximately universal.
\end{lem}
\begin{proof}
First of all,  $SNR(C_n\times C_n,\mu_n\times 0+0\times i\mu_n)=2SNR(C_n,\mu_n)$ and the rate of $(W_n\times W_n,\kappa_n\times\kappa_n,C_n\times C_n,\mu_n\times 0+0\times i\mu_n)$ is $2R_n$. So for distinct vectors $(v_1,v_2), (w_1,w_2)\in C_n\times C_n,$ we need to compute a lower bound for
the $i^{th}$ entry of $$((\mu_n(v_1)-\mu_n(w_1))^2+(\mu_n(v_2)-\mu_n(w_2))^2)/2 SNR_n,$$ and multiply over all $1\leq i\leq L$. If $v_1=w_1$, we see
that $\frac{1}{2^{2R_n+L+o(\log(SNR_n))}}$ is a lower bound for this product, which is of the form needed for the definition of approximately universal. A similar bound holds if $v_2=w_2$,
so now assume that $v_1\neq v_2$ and that $w_1\neq w_2$.
Then the arithmetic-geometric-mean inequality gives:
$$\prod_{1\leq i\leq L}{{(\mu_n(v_1)-\mu_n(w_1))^2+(\mu_n(v_2)-\mu_n(w_2))^2}\over{2 SNR_n}}\geq$$
$$\prod_{j=1,2}(\prod_{1\leq i\leq L}{{(\mu_n(v_j)-\mu_n(w_j))^2}\over {SNR_n}})^{1/2}
\geq \frac{1}{2^{2R_n+f(n)}},$$
where $f(n)$ is a function of $n$ that is $o(\log SNR_n)$ so is $o(\log 2 SRN_n)$.
\end{proof}

UDMs were constructed in \cite{Tal} because they can be used to build a sequence of  approximately universal codes. We will now show how to generalize
this construction.
Let $q$ be a power of a prime, $\mathbb{F}_q$ the field with $q$ elements, and $N$ a natural number. 

Let $\mathcal M=\{M_i|1\leq i\leq L\}$ be a collection of  $N\times N$ matrices with entries in $\mathbb{F}_q$. If $g$ is
a non-negative integer, we say that $\MC M$ is a set of (square) Universally Decodable Matrices of Genus $g$ (UDMG)
of length $L$
if for every $L$-tuple  $(\lambda_1,...,\lambda_L)$ of non-negative integers,
the matrix formed by concontenating the first $\lambda_i$ columns of $M_i$ is of full rank
whenever  $\displaystyle \sum_{i=1}^{L} \lambda_i \geq N+g$.

Assume now for every $N$ greater than some $N_0$, we have a UDMG $\mathcal M=\{M_i|1\leq i\leq L\}$ of length $L$ such that $L(N_0-g)\geq N_0$.
Let $K_i\in \mathbb{F}_q^N$ be the kernel of the linear transformation $\rho_i: \mathbb{F}_q^N\rightarrow  \mathbb{F}_q^N$ given by multiplication by $M_i$.
Since $M_i$ has rank at least $N-g$ by design, the dimension of $K_i$ is some $\delta_i\leq g$. Let $K$ be the span of $K_i$ for $1\leq i\leq L$,
and $W_N$ be a complementary space in $ \mathbb{F}_q^N$ to $K$, that is, $W_N+K=\mathbb{F}_{q^N}$ and $W_N\cap K=\{0\}$. Since the
dimension of $K$ is some $\Delta\leq \sum_{i=1}^L\delta_i\leq gL$, the dimension of $W_N$ is at least $N-Lg$.
 Note that by construction, $\rho_i$ is injective when restricted to $W_N$.
We then define $\kappa_N(v)=\{\rho_1(v),...,\rho_L(v)\}$ for $v\in W_N$, and
$C_n=\kappa_N(W_N)\subseteq \text{Mat}_{N\times L}( \mathbb{F}_q^N)$.
The modulation map $\mu_N:C_n\rightarrow \Bbb C^L$ will be the column-by-column extension of a map $\mu_0: \mathbb{F}_q^N\rightarrow \Bbb C$,
which we will now describe in detail. Given the result of Lemma~\ref{complex}, there is no reason 
not to build our example with $\mu_0$ being real-valued.

There is a standard map $p_{q^N}: \mathbb{F}_q^N\rightarrow \Bbb R$, built as follows. Arbitrarily identify $\mathbb{F}_q$ with $I_q=\{0,1,...,q-1\}$ and extend this
identification entry-by-entry from $\mathbb{F}_q^N\rightarrow I_q^N$. Now for any $a=(a_1,...,a_N)\in I_q^N$, let $p_{q^N}(a)=$
$$a_1q^{N-1}+\cdots+a_{N-1}q+a_N-{{q^N-1}\over {2}}$$ $$=(a_1-{{q-1}\over 2})q^{N-1}+\cdots+ (a_{N-1}-{{q-1}\over 2})q+(a_N-{{q-1}\over 2}).$$
This maps $I_q^N$ to $q^N$ unit-spaced points on the real line symmetrically placed about the origin. The map $p_{q^N}$ is standardly
called $q^N$-{\it PAM} (pulse-amplitude modulation). 
The modulation map we need to take is a weighted version of $q^N$-PAM.

We define $\mu_0(a_1,...a_N)=$
$$(a_1-{{q-1}\over 2})q^{N-1}w_1+\cdots +(a_{N-1}-{{q-1}\over 2})qw_{N-1}+(a_N-{{q-1}\over 2})w_N,$$
where  $w_i=(1+{{(q-1)(N+1-i)+1}\over{qN}})$ for $1\leq i\leq L$.
Note that $1\leq w_i\leq 2$. 

The reason for these weights is the following:
%We can now modulate $C$ into a gapped PAM constellation as in \cite{Tal} (one such constellation is a generalized Cantor set modulation, see \cite{CantorMod}) sitting in $\MB %R$, i.e. there is a map $\iota:C\rightarrow \MB R$ such that 
\begin{lem}\label{gap}
If two codewords $a=(a_1,...a_L)$ and $b=(b_1,...,b_L)$ in $C_N$ have $a_i=b_i$ for $i=1,...,m$ for some $1\leq m< L$, but $a_{m+1}\neq b_{m+1},$ then 
$$|\mu_0(a)-\mu_0(b)| >  q^{(N-m-1)}/N.$$
\end{lem}

\begin{proof} Without loss of generality, we can assume $\mu_0(a)>\mu_0(b)$, and then $\mu_0(a)-\mu_0(b)$ is minimized when $a_{m+1}=b_{m+1}+1$, and
$a_i=0,b_i=q-1$ for $m+2\leq i\leq N$. Hence $|\mu_0(a)-\mu_0(b)| \geq$ $$w_{m+1} q^{N-m-1}-(q-1)\sum_{k=1}^{N-m-1}w_{m+1+k}q^{N-m-1-k}=
1+q^{N-m-1}/N,$$
using the combinatorial identities $\sum_{i=0}^\ell x^i =(x^{\ell+1}-1)/(x-1)$ and $x$ times its derivative: $\sum_{i=1}^{\ell} ix^i={x\over {(x-1)^2}}(\ell x^{\ell+1}-(\ell+1)x^\ell+1).$
\end{proof}

Hence it is also appropriate to refer to $\mu_0$ as a {\it gapped} version of $q^N$-PAM, as they do in \cite{Tal}.

\begin{lem}\label{snr}
There are positive constants $\alpha$ and $\beta$ that depend on $q,g,$ and $L$, but not on $N$, such that,
$${{\alpha q^{2N}}\over N^2}\leq SNR(C_N,\mu_N)\leq \beta q^{2N}.$$
\end{lem}

\begin{proof}
First note that $SNR(C_N,\mu_N)=$
$${1\over |C_N|}\sum_{c\in C_N}||\mu_N(c)||^2={1\over |W_N|}\sum_{v\in W_N}\sum_{i=1}^L\mu_0(\rho_i(v))^2=\sum_{i=1}^L \sigma_i$$
where
$$\sigma_i={1\over {|W_N|}}\sum_{v\in W_N}\mu_0(\rho_i(v))^2={1\over {|Z_i|}}\sum_{z\in Z_i}\mu_0(z)^2,$$
where $Z_i=\rho_i(W_N)$.
%By assumption, the rank of $M_i$ is $N-\delta$ for some $\delta\leq g$. So if $Z_i$ is the image of the linear transformation from $\bF_q^N\rightarrow \bF_q^N$
%given by multiplication by $M_i$, then each element in $Z_i$ has an inverse image of size $q^\delta$. Hence

To get an upper bound on $\sigma_i$ we note:
$${1\over {|Z_i|}}\sum_{z\in Z_i}\mu_0(z)^2\leq {1\over {q^{N-Lg}}}\sum_{z\in I_q^N}\mu_0(z)^2$$
$$\leq {1\over{q^{N-Lg}}}\sum_{a\in I_q^N}((a_1-{{q-1}\over 2})q^{N-1}w_1+\cdots+ (a_{N-1}-{{q-1}\over 2})qw_{N-1}+(a_N-{{q-1}\over 2})w_N)^2$$
$$={1\over{q^{N-Lg}}}\sum_{a\in I_q^N}((a_1-{{q-1}\over 2})q^{N-1})^2w_1^2+\cdots+ ((a_{N-1}-{{q-1}\over 2})q)^2w_{N-1}^2+((a_N-{{q-1}\over 2}))^2w_N^2,$$
the sum of the cross terms vanishing because of the invariance of the set $\{a-(q-1)/2|a\in I_q\}$ under negation. Since $w_i\leq 2$, $\sigma_i$
is bounded above by
%$${4\over{q^{N-Lg}}}\sum_{a\in I_q^N}((a_1-{{q-1}\over 2})q^{N-1})^2+\cdots+ ((a_{N-1}-{{q-1}\over 2})q)^2+((a_N-{{q-1}\over 2}))^2$$
%$$={4\over{q^{N-Lg}}}(\sum_{a\in I_q}(a-{{q-1}\over 2})^2)q^{N-1}(q^{2N-2}+\cdots+ q^2+1).$$
%One can check that $\sum_{a\in I_q}(a-{{q-1}\over 2})^2=$ $$(q-1)q(2q-1)/6-q(q-1)^2/4=q(q^2-1)/12,$$ so
%$$\sigma_i\leq q^{Lg}q^{2N}/3.$$
%Summing over $1\leq i\leq L$, we get $SNR(C_N,\mu_N)\leq \beta q^{2N},$ where $\beta=Lq^{Lg}/3$.\\
%
%\color{red}
$$\frac{4}{q^{N-Lg}}\sum_{a\in I_q^N}((a_1-\frac{q-1}{2})q^{N-1})^2+\cdots+((a_N-\frac{q-1}{2}))^2\leq$$
$$ \frac{4}{q^{N-Lg}}
\sum_{a\in I_q^N} (\frac{q-1}{2})^2q^{2N-2}+\cdots+(\frac{q-1}{2})^2=\frac{4}{q^{N-Lg}}q^N(\frac{q-1}{2})^2(q^{2N-2}+\cdots+1)$$
$$=4q^{Lg}(\frac{q-1}{2})^2\frac{q^{2N}-1}{q^2-1}\leq q^{Lg}(\frac{q-1}{q+1})q^{2N}.$$
Summing over $1\leq i\leq L$, we get $SNR(C_N,\mu_N)\leq \beta q^{2N},$ where 
$\beta=Lq^{Lg}$.%(\frac{q-1}{q+1}).$
%\color{black}

The lower bound for $\sigma_i$ depends of the parity of $q$. First let us assume $q$ is odd, and set $z_0=({{q-1}\over 2},...,{{q-1}\over 2})\in I_q^N$. Then
$\mu_0(z_0)=0$, so
$$\sigma_i={1\over {|Z_i|}}\sum_{z\in Z_i}\mu_0(z)^2\geq {1\over {q^{N-\Delta}}}\sum_{z\in Z_i}|\mu_0(z)-\mu_0(z_0)|^2
\geq {1\over {q^{N-\Delta}}}\sum_{z\in Z_i} {{q^{2(N-m_0(z)-1)}}\over N^2},$$
by Lemma ~\ref{gap}, where $m_0(z)$ is the number of initial entries where $z$ and $z_0$ agree. Each addend decreases in size as $m_0(z)$ increases,
so if $Z'$ is the subset of elements in $I_q^N$ which agree with $z_0$ for their first initial $\Delta$ entries, then we have
$$\sigma_i\geq {1\over {q^{N-\Delta}}}\sum_{z\in Z'} {{q^{2(N-m_0(z)-1)}}\over N^2}=
{1\over {N^2q^{N-\Delta}}}(q-1)(q^{3(N-\Delta-1)}+\cdots q^3+1)$$ $$={{q-1}\over {(q^3-1)N^2q^{N-\Delta}}}(q^{3(N-\Delta)}-1)\geq 
{1\over {(3q^2)N^2q^{N-\Delta}}}q^{3(N-\Delta)}/2.$$
Summing over $1\leq i\leq L$, we see that when $q$ is odd we can take $\alpha=L/6q^{2gL+2}$.

When $q$ is even, $\mu_0(z)$ for $z\in I_q^N$ is minimized when $z$ is $z_0=(q/2+1,...,q/2+1)$ or $z_1=(q/2,...,q/2)$. Hence
$$\sigma_i={1\over {|Z_i|}}\sum_{z\in Z_i}\mu_0(z)^2\geq {1\over {q^{N-\Delta}}}\sum_{z\in Z_i}(\mu_0(z)^2-\mu_0(z_0)^2)=
$$ $$={1\over {q^{N-\Delta}}}\sum_{z\in Z_i}|\mu_0(z)-\mu_0(z_0)||\mu_0-\mu_0(z_1)|
\geq {1\over {q^{N-\Delta}}}\sum_{z\in Z_i} {{q^{(N-m_0(z)-1)+(N-m_1(z)-1)}}\over N^2},$$
by Lemma ~\ref{gap}, where $m_0(z)$ and $m_1(z)$ are respectively the number of initial entries where $z$ agrees with  $z_0$ and $z_1$. 
Note that either $m_0(z)$ or $m_1(z)$ vanishes, since $z_0$ and $z_1$ differ in all entries.
Again, each addend decreases in size as $m_0(z)$ or $m_1(z)$ increases,
so if $Z'$ is the subset of elements in $I_q^N$ which agree with $z_0$ or $z_1$ for their first initial $\Delta+1$ entries, then we have
$$\sigma_i\geq {1\over {q^{N-\Delta}}}\sum_{z\in Z'} {{q^{(N-m_0(z)-1)+(N-m_1(z)-1)}}\over N^2}=
{{q^{N-1}}\over {N^2q^{N-\Delta}}}2(q-1)(q^{2(N-\Delta-2)}+\cdots q^2+1)$$ $$={{2(q-1)q^{\Delta-1}}\over {(q^2-1)N^2}}(q^{2(N-\Delta-1)}-1)\geq 
{{q^{\Delta-1}}\over {(2q)N^2}}q^{2(N-\Delta-1)}.$$
Summing over $1\leq i\leq L$, we see that when $q$ is even we can take $\alpha=L/2q^{gL+4}$.

%Now $|Z_i|=q^{N-\delta}$ and 

%since $w_1>w_2>\cdots >w_L$,
%$$\sum_{z\in \{0\}^\delta\times I_q^{N-\delta}}|\mu_0(z)|^2\leq  \sum_{z\in Z_i}|\mu_0(z)|^2\leq \sum_{z\in I_q^{N-\delta} \times \{0\}^\delta}|\mu_0(z)|^2.$$
%Furthermore, since each $1\leq  w_i\leq 2$, $|p_{q^N}(z)|\leq |\mu_0(z)|\leq 2|p_{q^N}(z)|$. (does this work with minus signs?) 

%[ As a modulation map, it's signal-to-noise ratio $\sigma_{q^n}=SNR(I_q^N,p_{q^N})$ is
%$${1\over{q^N}}\sum_{a\in I_q^n}((a_1-(q-1)/2)q^{N-1}+\cdots+ (a_{N-1}-(q-1)/2)q+(a_N-(q-1)/2))^2=$$
%$${1\over{q^N}}\sum_{a\in I_q^n}((a_1-(q-1)/2)q^{N-1})^2+\cdots+ ((a_{N-1}-(q-1)/2)q)^2+((a_N-(q-1)/2))^2,$$
%the sum of the cross terms vanishing by the invariance of the set $\{a-(q-1)/2|a\in I_q\}$ under negation.
%Hence
%$$\sigma_{q^N}={1\over{q^N}}(\sum_{a\in I}(a-(q-1)/2)^2)q^{N-1}(q^{2N-2}+\cdots+ q^2+1).$$
%One can check that $\sum_{a\in I}(a-(q-1)/2)^2=(q-1)q(2q-1)/6-q(q-1)^2/4=q(q-1)(q+1)/12$, so
%$$\sigma_{q^N}=(q^{2N}-1)/12.$$

%Hence the SNR $S_{q^N}=SNR(I_q^N,\mu_0)$ of $\mu_0$ satisfies
%$$\sigma_{q^N}\leq S_{q^N}\leq 4 \sigma_{q^N}.$$

%Now suppose that $V \subseteq  I_q^n$ is a subset of size $q^{n-\delta}$ for some $\delta$. Then since $w_1\geq w_2\geq \cdots \geq w_N$, we have that
%SNR $\sigma_V$ attached to V is bounded above by the SNR $\alpha_{\delta,N}$ of when $V=I^{N-\delta}\times \{0\}^\delta$ and bounded below by the SNR $\beta_{\delta,N}$
%when $V=\{0\}^\delta\times I^{N-\delta}$. Note that $\beta_{\delta,N}=S_{q^{N-\delta}}$ and that $\alpha_{\delta,N}=q^{2\delta}\alpha_{\delta,N}$.]
\end{proof}

\begin{thm}
Fix $L>0, g\geq 0$. Suppose  for some $N_0$, for $N\geq N_0$ we have a sequence of UDMG $\mathcal M_N$ of genus $g$ of size $N\times N$ and length $L$,
with $L(N_0-g)\geq N_0$.
Then the corresponding sequence of codes  $(W_N,\kappa_N,C_N,\mu_N)$ built from $\mathcal M_N$ as above is approximately universal.
\end{thm}

\begin{proof}
Our proof is modeled on that in Appendix IV of \cite{Tal} .
%[with the slight difference that ours will be a proof]. %We consider the $I$ and $Q$ channels separately, these can be considered the real and imaginary parts of a complex number.
 %[and what are they?]. 
 %Suppose we have a code of rate $R$, so to send a codeword from this code we must have a
 %Suppose we have 
 %PAM constellation with $q^K$ points, where $q^K\approx 2^{R/2},$ since half of the rate is sent over each of the I and Q channels.
% [Why is $R$
%divided by 2 in the exponent?] 
%As in \cite{Tal} we have to use an irregularly spaced PAM constellation, so for each $m^{th}$ LSD (least significant digit) $q$-bit change there is a gap of $\gamma q^{m-1},$ where $\gamma$ is fixed.
Fix an $N\geq N_0,$ with $L(N_0-g)\geq N_0,$ and a UDMG  $\mathcal M=\{M_1,...,M_L\}$ of genus $g$ and length $L$ consisting of $N\times N$ matrices. Keep notation as above.
Let $v,w\in W_N$ be distinct, so for every $1\leq i\leq L$, $M_iv\neq M_iw$. Suppose that $M_iv$ and $M_iw$ agree in precisely the first $\lambda_i$  entries. Hence by Lemma~\ref{gap},
$$ |\mu_0(M_iv)-\mu_0(M_iw)| \geq q^{N-\lambda_i-1}/N,$$
so if $d_i=  |\mu_0(M_iv)-\mu_0(M_iw)|/\sqrt{SNR(C_N,\mu_N)}$, then by Lemma~\ref{snr}
$$d_i\geq q^{-\lambda_i-1}/N\sqrt{\beta}.$$
Since $\mathcal M$ is a UDMG of genus $g$, and since $v\neq w$, we must have  $\sum_{1\leq i \leq L} \lambda_i < N+g.$ Hence 
%Thus we need only consider the case  $\sum_{1\leq i \leq L} \lambda_i < K+g,$ in which case
\begin{align*}
\prod_{1\leq i\leq L} d_i ^{2} &\geq \prod_{1\leq i\leq L} q^{-2\lambda_i-2}/\beta N^2\\
&= \left(\frac{1}{\beta q^2 N^2}\right)^{L}q^{-2\sum_{i=1}^L\lambda_i}\\ & >  \left(\frac{1}{\beta q^2N^2}\right)^{L}q^{-2N-2g}\\ &=  \frac{1}{\beta^L q^{2L+2g}N^{2L} }q^{-2N}\\ &=
\frac{1}{\beta^Lq^{2L+2g}N^{2L}} \frac{1}{2^{2(R_N+\Delta)}},
\end{align*}
where recall $R_N=\log_2(W_N)$ is the rate of the code, and  $\Delta$ is the codimension of $W_N$  in $\mathbb{F}_{q^N},$ 
which is at most $gL$. Since 
$$\log{\beta^Lq^{2L+2g}N^{2L}2^{2\Delta}}=o(\log(\alpha q^{2N}/N^2)),$$ the Theorem follows from Lemmas~\ref{complex} and~\ref{snr}.
%Also, as in \cite{Tal}, we note that $q$ can increase as $\log(SNR)$ and UDMs will still be approximately universal, but if $q$ grows as $SNR^\mu$ for any $\mu>0$ then the %result does not hold. Note if we have rectangular UDMs of size $K\times N$ then we have unique decoding if $\sum_{1\leq i \leq r}\min(N, \lambda_i) \geq K+g$ and the above %proof follows mutatis mutandis.\\
%\indent In \cite{Tal} it is proved that the irregular spacing does not affect the multiplexing gain.
\end{proof}

%\indent 
The reason we include $N_0$ in our formulation is that we will show (see \S 5) that for fixed $N$ and $q,$ there is a bound for the number of parallel channels $L$ a message can be reliably sent over in terms of the genus $g$ of a UDMG.  As a result, allowing UDMGs (and not just only UDMs) offer new possibilities for coding design on slow-fading parallel channels, allowing for a larger value of $L$ for fixed $q$ and $N$.

\subsection*{Outline of the Paper}

\indent  %This paper is organized as follows.
In \S1 we give the abstract mathematical model of UDMG and derive their basic properties, including 
the vector-space realization of UDMG, equivalence of UDMG, and introduce sub- and quotient-UDMG. In \S2 we relate UDMG to linear vector codes, suggesting that the former is something of a generalization of the latter.
%gives the necessary mathematical background for the construction of UDMs of genus $g.$ A reader familiar with the Riemann-Roch theorem and some combinatorics will have the necessary tools to go straight to Section 4.
Section 3 gives our construction of UDMG of genus $g$ based on curves of genus $g$ (which we call ``Goppa UDMG"). 
%[t should be noted that our construction is similar to the construction of Sudan decodable con in \cite{RN}].
% It also relates our construction to  the generating matrix for the corresponding Goppa code. Indeed,
 In \cite{VG} (Proposition 14) they construct a UDM so that the matrix formed by concatenating the first row of each matrix in the UDM is the generator matrix for a Reed-Solomon code. Similarly in Theorem \ref{gen} we show that the matrix formed by concatenating the first column of each of the matrices in a Goppa UDMG  is the generating matrix for 
 a corresponding Goppa code.
 %where $\mathcal{C}$ is an nonsingular projective curve of genus $g$ and $\mathbf{P}$ is a set of $L$ rational points on $\MC{C}$ over $\mathbb{F}_q$ and $R$ is also a $%\mathbb{F}_q$-rational point on $\mathcal{C}$, so long as $R\notin \mathbf{P}$. In
In \S4  we provide an example of a Goppa UDMG worked out for a curve of genus $1$. 
The final  \S5 gives upper and constructable lower bounds on the number of matrixes in a UDMG in terms of its parameters.
%For information on Goppa codes see \cite{Step} or  \cite{JW}.\\

\section{Mathematical Model of UDMG}

%\indent Formally,  in \cite{VG} they define a UDM to be a set of matrices $$\mathbf{M}=\{M_1,\ldots , M_{L} \vert M_i \in \mathcal{M}_{K\times N}(\mathbb{F}_q)\}$$ such that if we %form the matrix $M$ by concatenating the first $\lambda_i$ columns of $M_i$, and if  $\displaystyle \sum_{i=1}^{L} \lambda_i \geq K$, then $M$ will be full rank for any such $\{\lambda_i\}_{i=1}^{L}$. %Moreover, in \cite{VG} they call a UDM as above an $(L,N,K,q)$-UDM.
\indent We first present the most general definition of Universally Decodable Matrices of genus $g$, and then specialize to the class of most interest in communications applications.
% [should we specify that the square ones are of the most interest, or just all $K_i$ the same?]

To fix notation, for a prime power $q$, let $\MB F_q$ be the field with $q$ elements, and for any $N,K>0$, let $\mathcal{M}_{K \times N}(\mathbb{F}_q)$ denote the $K\times N$
matrices with entries in $\MB F_q$. For any set $S$ of column vectors in $\MB F_q^K$, we let $\Span(S)$ be the span of $S$ over $\MB F_q$. We denote the $i^{th}$ entry of
a vector $v$ by $(v)_i$ and likewise denote the $j^{th}$-column of a matrix 
$M$ by $(M)_j$. All our vector spaces will be finite dimensional. We define an integer vector $\alpha$ to be greater than or equal to another integer vector $\beta$ of the same length, if every entry of $\alpha$ is greater than or equal to the corresponding entry of $\beta$. If $\eta\in \Bbb Z$, we let $\vec{\eta}$ denote the column vector all of
whose entries are $\eta$ and whose length is determined by context.  For any vector $\MBF N=(N_1,...,N_L)$ of integers, we will let $N=N(\MBF N)=\sum_{i=1}^L N_i$.

%\begin{defn}\label{allowable}
%Given a set of matrices $\Sigma=\{ S_1,...,S_h\}$ with S_i \in \mathcal{M}_{U\times V}(\mathbb{F}_q) \}_{i=0}^{h}$ we will say $\mathfrak{A}$ is a set of $T$  \textit{allowable %columns from} $\Sigma$  if there exists $\{\lambda_i \in \mathbb{Z}_{\geq 0}\}_{i=0}^{h}$ with 
 % $T=\sum_{i=0}^{h} \lambda_i$ and the first $\lambda_i$ columns of $S_i$ are in $\mathfrak{A}$.
  % \end{defn}

\begin{defn} \label{allowable}
For any positive integer $L$, let $\MBF N=(N_1,...,N_L)$ be a vector of non-negative integers. Fix $K>0, g \geq 0.$ 
Let $\mathbf{M}=\{M_1,\ldots, M_{L}\}$ be a set of $L$ matrices with $M_i \in \mathcal{M}_{K \times N_i}(\mathbb{F}_q).$ %\text{ and each } M_i \text{ is full rank}
For any $0\leq \lambda_i\leq N_i$ such that $\sum_{1=1}^L \lambda_i= K+g$, the collection $\mathfrak A$ of the first $\lambda_i$ columns from
each $M_i$ is called an \textit{allowable set of columns}  from $\MBF M$.

We say that $\MBF M$ is a (set of) Universally Decodable Matrices of genus $g$ (UDMG) if
 for every allowable set of columns $\mathfrak A$ of $\MBF M$, $\Span(\mathfrak A)=\MB F_q^K.$
 %is a set of matrices $$\mathbf{M}=\{M_1,\ldots, M_{L}\vert M_i \in \mathcal{M}_{K \times N_i}(\mathbb{F}_q) %\text{ and each } M_i \text{ is full rank}
%\},$$ such that if we take the first $\lambda_i$ columns of $M_i,$ and form $M$ by concatenating these columns, and if $\displaystyle \sum_{i=1}^{L} \lambda_i \geq K+g,$ then $M$ is full rank for any such $\{\lambda_i\}_{i=1}^{L}$. 
If so, we say that $\MBF M$ is a $(L,\MBF N,K,q,g)$-UDMG. The space of all UDMG with {\it parameters} $(L, \MBF N,K,q,g)$ will be denoted 
as $\mathcal{U}(L,\MBF N,K,q,g).$ We call the parameters $(L, \MBF N, K, q,g)$  respectively the {\it size}, {\it length},  {\it height},
{\it alphabet cardinality}, and {\it genus} of $\MBF M$.

If in addition there is a positive integer $\eta$ such that $N_i=\eta$, $1\leq i\leq L,$  $\MBF M$ will be called $\eta$-\textit{regular}, and the
set of such will be denoted by  $\MC{U}(L,\vec{\eta},K,q,g)$. 

\end{defn}

\begin{remark} It is only interesting to study UDMG $\MBF M$ which have at least one set of allowable columns, i.e.,
those for which $N\geq K+g$. Similarly, if any $N_i>K+g$, $(M_i)_j$ for $K+g<j\leq N_i$ is never an element of an allowable set of columns,
so we will only be interested in considering UDMG for which every $N_i\leq K+g$. Anomalous behavior occurs
when $K=1$, since then for any $g\geq 0$ and any $\MBF N$, we can have a code of unbounded size by taking each $M_i=\vec{1}$ (of length $N_i$).
\end{remark}

These considerations lead to the following:

\begin{defn} A UDMG $\MBF M\in \mathcal{U}(L,\MBF N,K,q,g)$ we be called {\it non-degenerate} if $N\geq K+g$, $N_i\leq K+g$ for
each $1\leq i\leq L$, and $K\geq 2$. The set of such will be denoted $\mathcal{U}_n(L,\MBF N,K,q,g).$ A UDMG which is not nondegenerate
will be called {\it degenerate}.
\end{defn}
%Note, much of what follows can be generalized to sets of matrices where $N$ is allowed to vary from one matrix to another in the set, i.e. $M_i\in \MC{M}_{K\times N_i}.$
\begin{remark}\label{contain}\label{allow}
%\begin{enumerate}
(1) We will be concerned throughout with the problems of finding {\it upper bounds} for $L$, by which we mean  $B_u=B_u(\MBF N, K, q,g)$ such that
$\mathcal{U}_n(L,\MBF N,K,q,g)$ is empty for $L>B_u$, and {\it constructable lower bounds} for $L$, by which we mean $B_\ell=B_\ell(\MBF N, K, q,g)$ such that
there exists an $L\geq B_\ell$ such that $\mathcal{U}_n(L,\MBF N,K,q,g)$ is non-empty.

%We also want $N\leq K$ since sending $K$ symbols over any one channel guarantees decodability. 

(2) In the definition of UDMG we do not require $g$ to be minimal, so for any $0\leq g\leq \tilde{g}$, 
$\mathcal{U}(L,\MBF N,K,q,g)\subseteq \mathcal{U}(L,\MBF N,K,q,\tilde{g}).$ However a non-degenerate UDMG with parameters $(L,\MBF N,K,q,g)$
is not necessarily a non-degenerate UDMG with parameters  $(L,\MBF N,K,q,\tilde{g})$

(3) Similarly, given any $\MC A=(A_1,...,A_L) \in\MC{U}(L,\MBF N,K,q,g)$, we can \textit{truncate} it to produce an  $\MC{A}'\in\MC{U}(L',\MBF N',K,q,g)$ for 
$\MBF N\geq \MBF N'\geq \vec{0}, $by taking $A'_i$ to be the first $N'_i$ columns of $A_i$ for all $1\leq i\leq L$. Here $L'$ is the number of non-zero  
$N'_i$ in $\MBF N'$. We call such an $\MC A'$ a {\it subUDMG} of $\MC A$.
(Taking $\MBF N'=\vec{0}$ produces what could only be called the {\it empty} UDMG.)
If each $N'_i<N_i$, we say that $\MC A'$ is a proper {\it subUDMG} of $\MC A$.

%[worry about redundancy of notation later].

%Any two UDMG gotten in this
%fashion from each other will be considered equivalent to each other.

%\item The union of two allowable sets of columns is allowable.
%\end{enumerate}

\end{remark}

Dual to the notion of subUDMG is taking
a quotient of a UDMG by a proper subUDMG. To explain this construction, it will be necessary to view UDMGs through a different guise.
Indeed, note that the definition of a UDMG considers the span of allowable columns of a set of matrices, and not the columns themselves. 
Hence it is sometimes useful to consider just the spans of the columns of a matrix in a UDMG, and not the columns themselves.  We build up the requisite notions as follows.

\begin{defn} Take $K,N>0,$ and $M \in \mathcal{M}_{K \times N}(\mathbb{F}_q).$
For any $1\leq j \leq N$, let $V(M)_j$ denote the span over $\MB F_q$ of the first $j$ columns of $M$.
We set $V(M)=\{V(M)_1,...,V(M)_N\}$ and call it the \textit{vector space realization} of $M$.

For any positive integer $N$, let $\MBF N=(N_1,...,N_L)$ be a vector of positive integers. % Fix $K>0.$ 
If $\mathbf{M}=\{M_1,\ldots, M_{L}\}$ is a set of $L$ matrices with $M_i \in \mathcal{M}_{K \times N_i}(\mathbb{F}_q)$, %\text{ and each } M_i \text{ is full rank
we set $V(\mathbf M)=\{V(M_1),....,V(M_L)\}$ and call it the \textit{vector space realization} of $\mathbf M$.

Note that all $V(M_i)_j$ are subspaces of $\MB F_q^K$. 
\end{defn}

\begin{defn}
%\begin{enumerate}
If $W$ is an $\MB F_q$-vector space, and $C: V_1,...,V_N$ is an ordered list of $N$ subspaces,  we call $C$ a \textit{chain} of subspaces
of $W$ if $V_1\subseteq \cdots \subseteq V_N$. We say the chain is \textit{closely nested} if $\dim(V_1)\leq 1$ and $\dim(V_{i+1}/V_i)\leq 1$ for each $1\leq i<N$.
\end{defn}
For any  $M\in \mathcal{M}_{K \times N}(\mathbb{F}_q),$ $V(M)$ is a chain of closely nested subspaces of $\MB F_q^K$. 
Conversely, given a chain $C: V_1\subseteq \cdots \subseteq V_N$ of closely nested subspaces of $\MB F_q^K$, one can form a matrix 
$M\in \mathcal{M}_{K \times N}(\mathbb{F}_q),$ such that $C=V(M)$ by setting $V_0=0$, and for each $0\leq i<N$ choosing $(M)_{i+1}\in \MB F_q^K$ to be a generator
of $V_{i+1}/V_i$ if the quotient is 1-dimensional, and arbitrarily in $V_{i+1}$ if $V_{i+1}=V_i$.

\begin{defn} We define a closely nested chain $C: V_1\subseteq \cdots \subseteq V_N$ of subspaces of an $\MB F_q$-vector space $W$ to be \textit{isomorphic} to a closely nested chain $C': V'_1\subseteq \cdots \subseteq V'_N$ of subspaces of an $\MB F_q$-vector space $W'$, if there is an $\MB F_q$-vector space isomorphism $\phi: W\rightarrow W'$ such that
$\phi(V_i)=V'_i$ for all $1\leq i\leq N$. 
\end{defn}

\begin{remark}
We extend this notion of isomorphism (element-by-element) to isomorphisms of ordered collections of closely nested chains of a vector space.
\end{remark}

With this we can now define two UDMGs to 
be \textit{isomorphic} if their vector space realizations are isomorphic ordered collections of closely nested chains of some $\MB F_q^K$.

%In particular, if $\MBF M=\{M_1,...,M_L\}$ is a UDMG with parameters $(L,\MBF N, K,q,g)$, and $A$
%is any invertible matrix in $\mathcal{M}_{K\times K}(\bF_q)$, then $\{AM_1,...,AM_L\}$ is an isomorphic UMDG, the isomorphisms on their
%vector space realizations being induced by the linear transformation $\bF_q^K\rightarrow \bF_q^K$ which multiplies vectors by $A$.

Given a set of matrices, one can test whether it is a UDMG in terms of its vector space realization.

\begin{defn} For any positive integer $L$, let $\MBF N=(N_1,...,N_L)$ be a vector of positive integers. Fix $K>0,$ and let $W$ be a vector space
over $\MB F_q$ of dimension $K$. For each $1\leq i\leq L$, let
$C_i: V^i_1\subseteq \cdots \subseteq V^i_{N_i}$ be a closely nested chain of subspaces of $W$, and $\MBF C=\{C_1,....,C_L\}$ be the ordered collection of these chains.
A vector $\Lambda=(\lambda_1,...,\lambda_L)$ of integers with $1\leq \lambda_i\leq N_i,$ such that $\sum_{i=1}^L \lambda_i\geq K+g$ is
called an \textit{allowable vector} for $\MBF C$.
We say that $\MBF C$ is a (set of) Universally Decodable Vector Spaces of genus $g$ (UDVSG) if
 for every allowable vector $\Lambda=(\lambda_1,...,\lambda_L)$ of $\MBF C$, the vector space sum
$\sum_{i=1}^L V^i_{\lambda_i}=W.$
 %is a set of matrices $$\mathbf{M}=\{M_1,\ldots, M_{L}\vert M_i \in \mathcal{M}_{K \times N_i}(\mathbb{F}_q) %\text{ and each } M_i \text{ is full rank}
%\},$$ such that if we take the first $\lambda_i$ columns of $M_i,$ and form $M$ by concatenating these columns, and if $\displaystyle \sum_{i=1}^{L} \lambda_i \geq K+g,$ then $M$ is full rank for any such $\{\lambda_i\}_{i=1}^{L}$. 
If so, we say that $\MBF C$ is a $(L,\MBF N,K,q,g)$-UDVSG attached to $W$. 
%The space of all UDVSG with parameters $(L, \MBF N,K,q,g)$ attached to
%$W$ will be denoted as $\mathcal{U}^W(L,\MBF N,K,q,g),$  though if $W=\MB F_q^K$, we omit it from the notation. 
\end{defn}

We have concocted these definitions so that the vector space realization of an  $(L,\MBF N,K,q,g)$-UDMG is a  $(L,\MBF N,K,q,g)$-UDVSG, and conversely,
that any  $(L,\MBF N,K,q,g)$-UDVSG attached to some $W$ is isomorphic to the vector space realization of some  $(L,\MBF N,K,q,g)$-UDMG.
Therefore the notion of UDVSM gives a coordinate-free way to study UDGMs. 
This is precisely what we need to make sense of quotients of a UDMG.

\begin{defn} We define a closely nested chain  $C: V_1\subseteq \cdots \subseteq V_N$ of subspaces of an $\MB F_q$-vector space $W$
to be \textit{reduced} if $V_1$ is non-trivial and {\it irredundant}
if it is reduced and $V_{i+1}\neq V_i$ for all $1\leq i<N$.  
We call a collection of closely nested chains of subspaces of $W$ to be be \textit{reduced} or {\it irredundant} if every chain is reduced or irredundant.
Likewise we call a collection of matrices to be \textit{reduced} or {\it irredundant} if its vector space realization is. 
\end{defn}

\begin{remark} (1) Any closely nested sequence of subspaces can be \textit{pruned} by removing any initial $0$-subspaces 
to make it reduced, and then be further pruned by removing any repeated subspaces to make it irredundant.
We can correspondingly \textit{prune} a matrix by removing any initial 0-columns or by removing any column in the span of the previous columns.

(2)  It is the regular, irredundant UDMG which are most important in the engineering application described in the Introduction
(where we considered only square UDMG for ease of exposition). The reason for irredundancy 
is that one would not waste power by transmitting a zero codeword or one known to be in the span of previous ones since we are assuming the channel 
transmits reliably what it does not erase. The reason for regularity is that each channel will be used for the same amount of time.

 We note that
the $\eta$-regular, irredundant UDMGs of size $L$, genus $0$, and height $K$ over $\bF_q$ comprise precisely the set $\mathcal{U}(L,\vec{\eta},K,q,0),$ 
which coincides with the space of all $(L,\eta,K,q)$-UDMs defined in \cite{VG}.

(3) In complete analogy to truncating a UDMG to form a subUDMG (or a proper subUDMG), one can truncate a UDVSG by truncating its chains
to form a subUDVSG (or proper UMVSG if every chain is truncated.). 

(4) Likewise we can define a UDVSG to be non-degenerate if it is the vector space realization of a non-degenerate UDMG.
\end{remark}

The following will be a fundamental notion for us. 
\begin{defn} Let  $C: V_1\subseteq V_2\subseteq \cdots \subseteq V_N$ be a chain of subspaces of an $\MB F_q$-vector space $W$. Let $B$ be any subspace of $W$.
We define the quotient chain $C_B$ of $C$ modulo $B$ to be the chain
$$(V_1+B)/B\subseteq \cdots \subseteq (V_N+B)/B,$$
of subspaces of $W/B$.
\end{defn}

\begin{prop}\label{nest} Let $W$ be an $\MB F_q$-vectors space and $B$ a subspace of $W$.
 If $C$ is a closely nested chain of subspaces of $W$, then $C_B$ is a closely nested chain of subspaces of $W/B$.
\end{prop}

\begin{proof} We just have to check that given two vector spaces $W_1\subseteq W_2$ with $\dim(W_2/W_1)=1$, then the dimension of
$V=((W_2+B)/B)/((W_1+B)/B)$ is at most 1. But $V$ is isomorphic to $(W_2/(W_2\cap B))/(W_1/(W_1\cap B))$ which has dimension
$\dim(W_2/W_1)-\dim((W_2\cap B)/(W_1\cap B))\leq 1$. 
\end{proof}

\begin{remark} If $C$ is a reduced  (or irredundant) chain, then in general, $C_B$ will not be reduced (or irredundant),
but one can of course prune $C_B$ to produce a reduced (or irredundant) chain.
%2) If $\MC M$ is a UDMG, and $V(M)$ is its vector space realizations as subspaces of some space $W$ of dimension $K$, then given any subspace $B$ of $W$ of dimension
%$K'$, then $V(A)_B$ is a collection of chains of subspaces of $W/B$, which for certain $B$, is actually isomorphic to the vector space realization of another UDMG. [fix or delete]
\end{remark}

\begin{thm}\label{quotientUDMG}
Let $\MBF C=\{C_1,...,C_L\}$ be a $(L,\MBF N,K,q,g)$-UDVSG  attached to an $\MB F_q$-vector space $W$ of dimension $K$.
Write the chain $C_i$ as $V^i_1\subseteq \cdots \subseteq V^i_{N_i}$. Let $\MBF S$ be a proper subUDVSG of $\MBF C$ of length $\MBF N'<\MBF N$.
%of length $ \mu=(mu_1,...,\mu_L)$, so $0\leq \mu_i<N_i$ for all $1\leq i\leq L$. 
Let $B$ be the vector space sum  
of all subspaces in the chains of $\MBF S$, which is $\sum_{i=1}^L V^i_{N'_i}$% for some $0\leq \mu_i \leq N_i-1$ 
(where we set $V^i_0=\{0\}$). 
Let $r=\max{(K-\sum_{i=1}^L  \MBF N'_i,0)}$.
%For every $1\leq i\leq L$ let 
%$0\leq \mu_i\leq N_i-1$ be such that , and let $B$ be the vector space sum  $\sum_{i=1}^L V^i_{\mu_i}$ (where we set $V^i_0=\{0\}$).
Then $\dim(B)=(K-r)-d$ for some $0\leq d\leq \min(K-r,g)$. Furthermore, let ${\MBF C}/\MBF S$ be $\{(C_1)_B,...,(C_L)_B\}$ with the first $\mu_i$ subspaces of each $(C_i)_B$ pruned, for each
$1\leq i\leq L$. Then ${\MBF C}/\MBF S$
is an  $(L,\MBF N-\MBF N', d+r,q, g-d)$-UDVSG,
which we call the {\it quotient} of $C$ by $S$.
\end{thm}

\begin{proof} First let us verify that if $r=\max{(K-\sum_{i=1}^L  \MBF N'_i,0)}$, and
$\dim(B)=(K-r)-d$, then $0\leq d\leq \min(K-r,g)$. First of all $d\geq 0$ if $r=0$ since $\dim(B)\leq \dim(W)=K$.
On the other hand, if $r>0$, by the definition of closely nested, $\dim{B}\leq \sum_{i=1}^L\MBF N'_i=K-r$.
Now we need to show $d\leq g$.
First of all, if $\MBF N'$ is an allowable vector, then $r=0$ and the dimension of $B$ is $K$, so $d=0$. 
Now suppose $\MBF N'$ is not an allowable vector, and for a contradiction, that $d>g$. Then there is an allowable vector $\MBF \lambda\geq \MBF N'$
with $\sum_{i=1}^L=K+g$, and so that
$\sum_{i=1}^L (\lambda_i-\MBF N'_i)\leq (K+g)-(K-r)=g +r$.
Hence $\dim(\sum_{1=1}^L (V^i)_{\lambda_i})\leq \dim(B)-r-d+g +r<K$, a contradiction of the definition of UDVSG.

Proposition~\ref{nest} gives that each $(C_i)_B$, $1\leq i\leq L$ is a closely nested chain of subspaces of $W/B$, which has dimension $d+r$.
 That the size of $\tilde{\MBF C}/\MBF S$ is $L$ follows from that each $\MBF N'_i\leq N_i-1$. So to
%Note that $d\leq g$, because if $d>g$, and $\lambda$ is an allowable vector greater than  $\mu$, then $\sum_{i=1}^L (\lambda_i-\mu_i)=(K+g)-(K-r)=g +r$, so $%\dim(\sum_{1=1}^L (V_i)_{\lambda_i})\leq \dim(B)+g +r<K$, a contradiction (of the definition of UDMG).
verify that  ${\MBF C}/\MBF S$ is a $(L,\MBF N-\MBF N', d+r,q, g-d)$-UDVSG  attached to $W/B$, we need just to take any vector $\lambda=(\lambda_1,...,\lambda_L)$,
with $\sum_{i=1}^L \lambda_i=(d+r)+(g-d)=g+r$, and check that $\sum_{i=1}^L (V^i_{\MBF N'_i+\lambda_i}+B)/B)=W/B$. But this follows since
$\sum_{i=1}^L(\lambda_i+\MBF N'_i)\geq K+g$, so  $\sum_{i=1}^L V^i_{\MBF N'_i+\lambda_i}=W$. 
\end{proof}

\begin{remark} 1) If $S$ is the empty UDVSG, then $\MBF C/\MBF S=\MBF C$. 

2) Even if $\MBF C$ is nondegenerate, it can happen  that $\tilde{\MBF C}/\MBF S$ is not.
%Indeed, while $\MC C$ being nondegererate means $N\geq K+g$, so $\sum(N_i-\mu_i)=N-(K-r)\geq g+r=(d+r)+(g-d)$, it may be the case that some 
%$N_i-\mu_i<g+r$, or that $d+r<2$. 
\end{remark}

\begin{defn} The quotient of a UDMG $\MBF C$ by a proper subUDMG $\MBF S$ is a UDMG isomorphic to the quotient of $V(\MBF C)$ by $V(\MBF S)$
(so is only defined up to isomorphism).
\end{defn}

%\vspace{3mm}

\section{Relationship between UDMGs and linear vector codes}

%[From hereonin, should $\MBF C$ be $\MBF M$?]

\indent If $C$ is an $[n,k,d]$ $\MB{F}_q$-linear vector code (that is, a $k$-dimensional subspace of $\MB F_q^n$ whose minimal Hamming distance is $d$), 
then the Singleton Bound states that  $n+1-d-k\geq 0$ see \cite{Tsfa} or \cite{JW}. We standardly call $s=n+1-d-k$
the {\it Singleton Defect} of $C$ \cite{codedefect}. If $s=0$ then $C$ is a maximal distance separating \TI{MDS code}. More generally,
linear codes of defect $s$ are called A$^s$MDS codes

\begin{prop}\label{linear} Let $\MBF M=\{M_i\}_{1\leq i\leq L}$ be an $(L,\MBF N,K,q,g)$ UDMG, and $G$ be the $K\times L$ matrix whose $i^{th}$-column
is the first column of $M_i$. Then if $L\geq K+g$, $G$ is the generating matrix for some $\MB F_q$-linear $[L,K,d]$-code $\tilde{ C}$ of defect at most $g$. 
In particular, if  
$g=0$, $\tilde{C}$ is an $MDS$-code.
\end{prop}

\begin{proof} Since $L\geq K+g$ and $\MBF M$ is a UDMG of genus $g$, $G$ has rank $K$ over $\bF_q$. We just have to bound the minimum distance of $\tilde{C}$.
If $v\in \tilde{C}$ is non-zero, there is an invertible $K\times K$ matrix $M$ over $\bF_q$ such that $v$ is a row of $MG$. Since $G$ is a matrix whose every $K\times (K+g)$ minor has rank $K$, the same is true of $MG$. Hence $v$ has at most $K+g-1$ zero entries, so has Hamming weight at least $L+1-K-g$. Hence the
Singleton defect of $\tilde{C}$ is at most $g$. 
\end{proof}

Bounds on the size of MDS codes have been extensively studied (they are the subject of the famed ``MDS"-conjecture), and still comprise an active area of research. 
Although there are sporadic better results, the best known bound for a general $[n,k,d]$ $\MB{F}_q$-linear MDS code is that $n\leq k+q-1$ (see \cite{codedefect} or \cite{Tsfa}). 

We will make use of the generalization of this bound to A${^s}$MDS codes:
%There is also a known bound for $[n,k,d] \ \MB{F}_q$-linear codes of defect $s$ called A$^s$MDS codes. 
\begin{lem}\label{linearbound}\cite{codedefect}
Let $C$ be a $[n,k,d]$ $\MB{F}_q$-linear code of Singleton defect $s$. Then $$n \leq k-2+(q+1)(s+1).$$
\end{lem}

Cognizant of Proposition~\ref{linear}, 
in \cite{VG}, for $L\geq K$, they construct a UDMG $\MBF M$ with parameters $(L,\vec{K},K,q,0)$ whose associated linear vector code $\tilde{ C}$ is a Reed-Solomon code with parameters 
$[L,K,L-K+1]$. The Reed-Soloman codes are the classic non-trivial example of MDS codes. In the next section we generalize this construction to more generally build
UDMG $\MBF M$ with parameters $(L,\vec{K},K,q,g),$ whose associated linear vector codes $\tilde{ C}$ are Goppa codes constructed from curves of genus $g$ over
$\bF_q$, and have parameters $[L,K,d]$ for some $d\geq L-K+1-g$, so have Singleton Defect $s\leq g$ (see Remark~\ref{defect},
Theorem~\ref{goppa}, and Theorem ~\ref{gen}).
% [Why $s$ and not $d$? Does this inequality go the right way?]

We cannot help from noting that since every linear code over $\bF_q$ is an A${^s}$MDS codes for some $s$, there is a sense in which UDMG are generalizations
of linear vector codes over $\bF_q$. This leads to the tantalizing question of what the dual of a UDMG should be, and what properties it would have.
Likewise, is there a good notion of what the spectrum of a UDMG should be?
% Let $\MBF{M}$ be a UDM and let $C$ be the corresponding Reed-Solomon code formed from the top rows of $\MBF{M}$. Let $L=\#\MBF{M}$ and $C\subset \MB{F}_q^K.$ 

% Let $C$ be a $[n,k,d]$ linear code. Then $S(C)=d-n+k-1$ is called the Singleton defect of $C$, i.e. $S(C)$ measures how far $C$ is from the Singleton bound, see  Clearly, MDS %codes have defect $0$.
.

\section{Goppa UDMGs}

Everything we require on the theory of curves over finite fields and their associated Goppa codes can be found in \cite{JW}. We will recall what we need by way of
establishing notation.

%If the reader is familiar with the Riemann-Roch theorem and the combinatorics of vector space containment they can skip this section.

%\begin{notation}
%We will denote the $j^{th}$ column of a matrix by $(M)_j.$
%\end{notation}
By a {\it curve} $X$ over $\bF_q$ we mean a one-dimensional non-singular projective variety (always taken to be irreducible)
over the algebraic closure $\bar{\bF}_q$ of $\bF_q$ which is defined over $\bF_q$
We will let  $X(\bar{\bF}_q)$denote the points of $X$ defined over $\bar{\bF}_q$ and $X(\bF_q)$ be the subset of points defined over $\bF_q$. 
A {\it divisor} $D$ on $X$ is an element of the free abelian group generated by $X(\bar{\bF}_q)$, 
so can be written as $D=\displaystyle \sum_{P\in X(\bar{\bF}_q)}n(P)P$, where all but finitely many $n(P)\in \bZ$ vanish. The set of $P$ for which $n(P)\neq 0$
is called the {\it support} of $D$ and written as supp$(D)$. We write $\deg(D)$ for the {\it degree} of $D$,
which is $\displaystyle \sum_{P\in X(\bar{\bF}_q)}n(P)$. We put a partial order on divisors by saying $D\geq 0$ if each $n(P)\geq 0$.
Let $\bar{\bF}_q(X)$ and $\bF_q(X)$
respectively denote the field of functions on $X$ and the subfield of functions defined over $\bF_q$. 
For every $P\in X(\bar{\bF}_q)$ we let $v_P$ be the discrete valuation on $\bar{\bF}_q(X)$
that measures the order of zero (or pole) at $P$ of a function. 
To every non-zero $f\in \bar{\bF}_q(X)$ we can associate a
divisor $(f)=\displaystyle \sum_{P\in X(\bar{\bF}_q)}v_P(f)P$. Likewise, if $\omega$ is a differential on $X$, 
for every $P\in X(\bar{\bF}_q)$, we can let $t_P$ be a uniformizer in the valuation ring of $\bar{\bF}_q(X)$ associated to $v_P$,
and define $v_P(\omega)=v_P(\omega/dt_P)$, which is independent of the choice of $t_P$. 
Then we define the divisor of $\omega$ to be
$(\omega)=\displaystyle \sum_{P\in X(\bar{\bF}_q)}v_P(\omega)P$. We put an equivalence relation on divisors by saying that $D_1$ and $D_2$ are
{\it linearly equivalent} if $D_1-D_2$ is the divisor of a function: if so we write $D_1\sim D_2$. For any differential
$\omega$ we set $\kappa=(\omega)$ which is called a canonical divisor of $X$, which is well-defined up to linear equivalence since the ratio of any two differentials is a function. 
\begin{defn} Let $D$ be a divisor on $X$. Let
\label{LD}
 $$\mathcal{L}(D)=\{ f \in \bar{\bF}_q(X)-\{0\} | (f) \geq \ -D\} \cup \{0\}.$$
\end{defn}

% It should also be noted that if $D=nP$ then the definition of  $\mathcal{L}(nP)$ implies that any non-zero function $f \in  \mathcal{L}(nP)$ has at most a pole of order $n$ at $P$ and no other poles. Since $f\in K(\mathcal{C})^*$ it follows that $f$ can have at most $n$ zeros when $f \in \mathcal{L}(nP)$. Moreover,
The space  $\mathcal{L}(D)$ is a finite dimensional $\bar{\bF}_q$-vector space,  and we let $\textit{l}(D)$ denote its dimension.
If $D$ is defined over $\bF_q$ (i.e., is fixed by the Galois group of $\bar{\bF}_q$ over $\bF_q$),
then $\mathcal{L}(D)$ has a basis that lies in $\bF_q(X)$.

For any divisor $D$ and point $P\in X(\bar{\bF}_q)$ not in the support of $D$, we can define an {\it increasing zero basis at $P$} for $\mathcal{L}(D)$ to be an ordered basis
$(f_1,...,f_{\ell(D)})$ such that for all $1\leq i<\ell(D)$, $v_P(f_{i+1})>v_P(f_i)$. (One can also do decreasing pole bases.)
Such bases exist because $v_P$ is a discrete valuation and $\bar{\bF}_q$ is the redisue field of $v_P$. If $D$ and $P$ are defined
over $\bF_q$, the increasing zero basis can be taken to have elements in $\bF_q(X)$, in which case we call it {\it an increasing zero basis at $P$ over $\bF_q$.}
%From the definition of $\mathcal{L}(D)$ we have the inclusions of vector spaces $\mathcal{L}(D) \supseteq \mathcal{L}(D-P) \supseteq \ldots \supseteq \mathcal{L}(D-JP)$ for any %positive integer $J$ and this shows that given a point $P$ there exists an increasing zero basis of $\mathcal{L}(D)$ about $P$.

Every non-zero function on $X$ has the same number of poles and zeros, so a function without a pole is a constant, and has a trivial divisor. In other words:
\begin{lem}\label{neg}
If $\deg(D)<0$ then \textit{l}($D)=0$. Likewise, $\mathcal{L}(0)=\bar{\mathbb{F}}_q$ so $\mathit{l}(0)=1.$
\end{lem}
%For a proof see \cite{Sil}.\\
%\indent Let $\mathcal{}$ be a genus $g$ curve. %It should be noted that \ref{neg} implies if we consider a function $f\in \mathcal{L}(KP_{\infty}-\lambda_0P_0-\ldots -\lambda_{r-1}P_{r-1})$ 
%and if $\displaystyle\sum_{i=0}^{r-1} \lambda_i \geq K+1$ then $f=0$.

Fundamental to the subject is the Riemann-Roch Theorem.
\begin{thm}
\label{Riemann} (Riemann-Roch) For any curve $X$ there is a non-negative integer $g$ called its {\it genus}, 
such that for any canonical divisor $\kappa$ on $X$, and any divisor $D$,
$$\textit{l}(D)-\textit{l}(\kappa-D)=\mathrm{deg}(D)-g+1.$$
\end{thm}
%For details see \cite{Sil} or \cite{Shaf}. 
Note that setting $D=0$ gives $\mathit{l}(\kappa)=g$. Then setting $D=\kappa$ gives that $\deg \kappa =2g-2.$

\begin{cor}
\label{RRcor}
It now follows from Lemma \ref{neg} that if $\textrm{deg}(D)>2g-2,$ then $\textit{l}(D)=1-g+\textrm{deg}(D)$.
\end{cor}

\begin{defn} \cite{JW} Let $X$ be a curve over $\bF_q$ and $\mathbf{P}=\{P_1,\ldots , P_n\}\subseteq X(\bF_q)$. Let $D$ be a divisor on $X$ 
over $\bF_q$,
% such that  $2g-2\leq \deg(D) \leq a$
with $\text{supp}(D) \cap \mathbf{P} =  \varnothing$ . 
Then $\mathbf{C}(X,\mathbf{P},D)=\{(f(P_1),\ldots , f(P_n)) \vert f\in \mathcal{L}(D) \}$ is the \textit{Goppa code} 
associated with $(X,\mathbf{P},D)$ and has parameters $[n,\mathit{l}(D)-\mathit{l}(D-\MBF P),d],$ for some $d\geq n-\deg(D),$ 
as an $\MB F_q$-linear vector code.
%[What are the parameters?]
\end{defn}

\begin{remark}\label{defect}  Note in particular that if $n>\textrm{deg}(D)$, then the parameters simplify to 
$[n,\mathit{l}(D),d],$ so by the Riemann-Roch Theorem, the Singleton defect of $\mathbf{C}(X,\mathbf{P},D)$ is
less than or equal to $g$.
\end{remark}
%kor more information see \cite{JW}.

%kbegin{lem}
%\label{hyper}
%Let $n\in\mathbb{Z}_{\geq 2}.$
% Then the number of hyperplanes in $\mathbb{F}_q^n$ of co-dimension $1$
 % is equal to the number of points in $\mathbb{P}^{n-1}(\mathbb{F}_q)$.
 % Moreover, the number of points in $\mathbb{P}^{n-1}(\mathbb{F}_q)=\frac{q^n-1}{q-1}.$
%\end{lem}

%This is well known.

%\begin{lem}
%\label{sub}
%Exactly $q+1$ distinct $j$ dimensional subspaces of $\mathbb{F}_q^{j+1}$ contain any fixed $j-1$ dimensional subspace $A$.
%\end{lem}
%\begin{proof}
%Fix a $(j-1)$ dimensional subspace $A$. 
%Then $\mathbb{F}_q^{j+1}/A \cong \mathbb{F}_q^2$ by the canonical quotient map. Note under the quotient map each subspace of dimension $j$ containing $A$ in $
%\mathbb{F}_q^{j+1}$ corresponds to a  $1$ dimensional subspace in $\mathbb{F}_q^2$. The number of distinct lines in $\mathbb{F}_q^2$ is the number of points in $
%\mathbb{P}^1(\mathbb{F}_q)$ which is $q+1$. The result follows.
%Then there are $q^{j+1}-q^{j-1}$ vectors in $\mathbb{F}_q^{j+1} - A$ and if $B$ is a $j $ dimensional subspace containing $A$ then there are $q^j - q^{j-1}$ vectors in $B-A$. %Thus there are $\displaystyle \frac{q^{j+1}-q^{j-1}}{q^j - q^{j-1}}=q+1$ distinct $j$ dimensional subspaces that contain $A$.
%\end{proof}

%\section{Goppa UDMs}
We now have what we need to construct our Goppa UDMGs. Our construction directly generalizes the one given in \cite{VG}, once their results on bivariate polynomials
are reinterpreted in terms of statements about the arithmetic and geometry of the projective line $\mathbb{P}^1$ over $\bF_q$.
%. In \cite{VG} the authors took an information vector of length $K$ and built a corresponding function in the $\mathbb{F}_q$-vector space spanned by $\{Z^{K-1},XZ^{K-2}, \ldots, %X^{K-1}\}$, which we'll call the standard basis for homogenous polynomials in the $2$ variables $X,Z$ of degree $K-1$.   When a basis has increasing zero orders at a point, we'll %call it an increasing zero basis as in \cite{RN}. The standard basis above is an increasing zero basis for the point $[0,1]\in \mathbb{P}^1(\mathbb{F}_q).$ Then the authors of 
%\cite{VG} found different bases for this same vector space, an increasing zero basis expanded about each point of $\mathbb{P}^1(\mathbb{F}_q)$. The resulting $\mathcal{A}\in
%\mathcal{U}(L,\{N\},K,q,0)$ they construct is the set of change of basis matrices between each basis and the standard basis. It should be noted that $\mathbb{P}^1(\mathbb{F}_q)$ %is a genus $0$ curve. Thus one way to generalize this construction is to do the analogous construction for higher genus curves.\\
% Note that for arbitrary $\MBF N,K$ it appears to be non-trivial to produce UDMGs of genus $g$ without using the construction given below. 
Our construction produces an $\mathcal{A}\in\mathcal{U}(L,\vec{K},K,q,g)$ but we can always truncate this to an $\tilde{\mathcal{A}}\in\MC{U}(L,\MBF N,K,q,g)$ as in Remark \ref{contain} where each $N_i\leq K$.

\begin{thm}\label{goppa}  Let $X$ be a curve of genus $g$ over $\bF_q$, and fix any $K>g-1.$ Let $a=K+g-1$, $D$ be a divisor of degree $a$ on $X$ defined over $\MB F_q$, and $\MBF{P}=\{P_1,...,P_L\}\subseteq X(\bF_q)$ be such that $\Supp D\cap \MBF{P}=\varnothing$.
Let $B_0$ be any ordered basis for $\mathcal{L}(D)$ as an $\bF_q$-vector space, and for $1\leq i\leq L,$ let $B_i$ be an increasing zero basis for $\MC L(D)$ at $P_i$
 over $\bF_q$. For $1\leq i\leq L,$  let $M_i$ be the change-of-basis matrix from $B_0$ to $B_i$, which is a $K\times K$ matrix. Then
the set $\MBF M=\{M_1,...,M_L\}$ is a UDMG with parameters $(L,\vec{K},K,q,g)$.
We will call $\MBF M$ the Goppa UDMG associated with $(X,\mathbf{P},D),$ and it is nondegenerate if $L,K\geq 2$.
\end{thm} 

\begin{proof} For the size of each $M_i$,  note that $K>g-1$ implies that $ a>2g-2$, so  $\ell(D)=a-g+1=K$ by Corollary \ref{RRcor}.
If we write $B_i=\{B_{ij}\} ,$ $1\leq j \leq K$, then it follows by induction that $v_{P_i}(B_{i,j+1})\geq j$ since
\label{induct}$$v_{P_i}(B_{i,j+1})> v_{P_i}(B_{ij})\geq j-1,$$
for $1\leq j<K$. 
Since $\Supp D\cap \MBF{P}=\varnothing$ we lose no generality by  taking $B_{i1}(P_i)=1$. By construction, writing $B_i$ as column vectors, we have
\begin{equation}\label{induct} M_iB_i=B_0,\end{equation} for each $1\leq i\leq K$.
To prove the theorem, we must verify that given any allowable set of columns $0\leq \lambda_i\leq K$ such that $\sum_{i=1}^L\lambda_i=K+g=a+1$, that if $\mu_i$ is the
$K\times \lambda_i$ matrix consisting of the first $\lambda_i$ columns of $M_i$, and $M$ is the
$K\times (K+g)$ matrix formed by concatenating $\mu_i$ for $1\leq i\leq L$, then $M$ has rank $K$. Equivalently, we need to show
that every row vector $u$ of length $K$ with entries in $\bF_q$ in the left-nullspace of $M$ is the zero vector. 
%Let $B\in (\bF_q)^{K+g}$ be formed by 
%stacking the first $\lambda_i$ entries of $B_i^t$ for $1\leq i\leq L$, and let $u_i$ be the vector of length $\lambda_i$ with entries in $\bF_q$ such that $u$
%is the concatenation of $u_i$ for $1\leq i\leq L$. 
Note that  $uM =0$ implies  $uN_i=0, $
for all $1\leq i\leq L$. 
Set $f=uB_0$, which is the zero function precisely when $u$ is the zero vector. By (\ref{induct}), $f=uM_iB_i$ for each $1\leq i\leq L$.  
However,
%$\vec{u}MB=0$, so for each $1\leq i\leq L$, $\vec{u}(\text{first }\lambda_i \text{ columns of }M_i)=0.$ 
$uN_i=0$ then implies that
%$$\text{Hence } f=(0,\ldots, 0, c_{\lambda_i+1},\ldots ,c_{\textit{l}(D)})B_i^t \text{ for some }\lambda_i \implies 
$v_{P_i}(f)\geq v_{P_i}(B_{i,\lambda_i+1})\geq \lambda_i.$
 Thus $f\in \mathcal{L}(E),$ where $E=D-\lambda_1P_1-\ldots-\lambda_{L}P_{L}$. But 
 $\mathrm{deg}(E)<0$ and so \textit{l}$(E)=0$ by Lemma \ref{neg}. 
Hence $f=0$, $u=0$, and $M$ is of full rank. Thus  $\MBF M   \in\mathcal{U}(L,\vec{K},K,q,g)$.
\end{proof}
 
% By the similar arguments as above in the genus $1$ case we get that $f\in \mathcal{L}(sP_{\infty}-\lambda_0P_0-\ldots-\lambda_{r-1}P_{r-1}).$  Thus $f=0
% $ by \ref{neg}. Hence, the set of change of basis matrices $\mathbf{M}=\{M_0,\ldots, M_{r-1}\vert M_i \in \mathcal{M}_{K \times K}\}$ is a UDM of genus $g$. \\

 \begin{thm}
 \label{gen}
With notation as in Theorem \ref{goppa}, the matrix formed by concatenating the first column of each $\{ M_i \}$, $1\leq i\leq L$, is the generating matrix for the Goppa code associated with
$(X,\mathbf{P},D)$.
\end{thm}
\begin{proof}
 A generating matrix for the Goppa code with parameters $(X,\mathbf{P},D)$ has entries $g_{ij}=f_i(P_j)$, for $1\leq i\leq K$, $1\leq j\leq L$,
  where $\{ f_i\}$, $1\leq i\leq K,$ is any basis for $\mathcal{L}(D)$ defined over
 $\bF_q$. In particular, we can take $f_i=B_{0i},$ $1\leq i\leq K$, as this basis.
 
Recall that $M_j$ is defined to be the matrix satisfying $M_j B_j=B_0,$ for $1\leq j\leq K$. Thus we have  $M_j B_j(P_j)=B_0(P_j).$ But $B_j$ is an increasing zero basis at $P_j$, so $B_{jk}(P_j)=0$ for $2\leq k \leq L$. We took $B_{j1}(P_j)=1$, so $B_{0i}(P_j)=(M_j)_{i,1}$ as desired.
\end{proof}

\section{An example of a Goppa UDMG of genus 1}
\begin{example}

\label{easy}
Let $X$ be the curve in $\MB P^2$ defined by the equation $S^2T=R^3+RT^2+T^3$ over $\mathbb{F}_5$. Since the cubic is nonsingular over $\bar{\MB F}_5$,
 $X$ is  a  nonsingular projective curve of genus $1$ over $\bF_5$~\cite{Sil}. Set
 $$\mathbf{P}=\{P_1,P_2,P_3,P_4,P_5,P_6,P_7,P_8,P_9\}=$$
 $$
 \{[0,1,1], [4,2,1], [3,4,1], [0,4,1], [4,3,1], [3,1,1],[2,1,1],[2,4,1],[0,1,0]\},$$ which is all of  $X(\bF_5)$,
  and let $D$ be the degree $3$ divisor cut out by the hyperplane $S+R=0$ on $X$. Since the hyperplane is defined over $\bF_5$, the same is
  true of $D$, and one checks that none of the points in $\MBF P$ are in the support of $D$. 
%Suppose we want to send a vector of length $3$ with entries in $\mathbb{F}_5.$ Then,
%Set $K=3$; then t
To build the Goppa UDMG associated with $(\mathcal{C},\MBF P,D)$, we need to calculate an increasing zero basis for $\mathcal{L}(D)$ about each point in $\mathbf{P}.$ 
%[And a decreasing pole basis at $P_0$.]  
Let $Q=[0,1,0]$.
If $r=R/T,s=S/T\in \bF_q(C)$, then a standard fact about genus 1 curves give that $1,r,s$ span $\mathcal{L}(3Q)$, and that the divisor of
$r+s$ is $D-3Q$ (see \cite{Sil}, Chapter 3). Hence $\alpha=1/(r+s)$, $\beta=r/(r+s)$, and $\gamma=s/(r+s)$ span $\mathcal{L}(D)$.
%We will let $\{1,x,y\}=B_0$. 
%Let $\{P_1, \ldots, P_9\}=\mathbf{P}$ represent the same list of points as above in the same order, and let
Let
 $B_i$, $1\leq i\leq 9$, be an increasing zero basis for $\mathcal{L}(D)$ about the point $P_i$. Then we can take:
% $$B_1=(Z^3, XZ^2, YZ^2-3XZ^2-Z^3), \  \ B_2=(Z^3, XZ^2-4Z^2, YZ^2-XZ^2+2Z^3), $$ $$ B_3=(Z^3, XZ^2-3Z^3, YZ^2-XZ^2+4Z^3),\  \ B_4=(Z^3, XZ^2, YZ^2-2XZ^2-4Z^3), $$ 
 %$$B_5=(Z^3,XZ^2-4Z^3, YZ^2-4XZ^2-2Z^3), \  \  B_6=(Z^3, XZ^2-3Z^3, YZ^2+XZ^2-4Z^3),$$ $$ B_7=(Z^3, XZ^2-2Z^3, YZ^2-4XZ^2+2Z^3),\  \ B_8=(Z^3, XZ^2-2Z^3, YZ^2-%XZ^2+Z^3).$$  
%$$B_1=(1, x, y-3x-1), \  \ B_2=(1, x-4, y-x+2), $$ $$ B_3=(1, x-3, y-x+4),\  \ B_4=(1, x, y-2x-4), $$ 
 %$$B_5=(1,x-4, y-4x-2), \  \  B_6=(1, x-3, y+x-4),$$ $$ B_7=(1, x-2, y-4x+2),\  \ B_8=(1, x-2, y-x+1),$$
 %$$\text{ and }B_9=(y,x,1).$$
 $$B_1^t=(\alpha, \beta, \gamma-3\beta-\alpha), \  \ B_2^t=(\alpha, \beta-4\alpha, \gamma-\beta+2\alpha), $$ $$ B_3^t=(\alpha, \beta-3\alpha, \gamma-\beta+4\alpha),\  \ B_4^t=(\alpha, \beta, \gamma-2\beta-4\alpha), $$ 
 $$B_5^t=(\alpha,\beta-4\alpha, \gamma-4\beta-2\alpha), \  \  B_6^t=(\alpha, \beta-3\alpha, \gamma+\beta-4\alpha),$$ $$ B_7^t=(\alpha, \beta-2\alpha, \gamma-4\beta+2\alpha),\  \ B_8^t=(\alpha, \beta-2\alpha, \gamma-\beta+\alpha),$$
 $$\text{ and }B_9^t=(\gamma,\beta,\alpha),$$
 where the superscript $t$ denotes taking the transpose.
 % be a decreasing pole basis at $P_0$.]
 Let $B_0=B_1$ and let $M_i$ be the change of basis matrix that satisfies $$M_iB_i=B_0.$$  Then we get: $$M_1=\left(
\begin{array}{ccc}
1 & 0 & 0 \\
0 & 1 & 0 \\
0 & 0 & 1\\
\end{array}
\right)
M_2=
\left(
\begin{array}{ccc}
1 & 0 & 0\\
4 & 1 & 0\\
4 & 3 & 1\\
\end{array}
\right)
M_3=\left(
\begin{array}{ccc}
1 & 0 & 0\\
3 & 1 & 0\\
4 & 3 & 1\\
\end{array}
\right),
$$ $$
M_4=\left(
\begin{array}{ccc}
1 & 0& 0\\
0 & 1 &0\\
3 & 4 & 1\\
\end{array}
\right)
M_5=\left(
\begin{array}{ccc}
1 & 0 &0\\
4 & 1& 0\\
0 & 1 &1\\
\end{array}
\right)
M_6=\left(
\begin{array}{ccc}
1 & 0 &0\\
3 & 1& 0\\
1 & 1 &1\\
\end{array}
\right),
$$ 
$$
M_7=\left(
\begin{array}{ccc}
1 & 0 &0\\
2 & 1& 0\\
4 & 1 &1\\
\end{array}
\right)
M_8=\left(
\begin{array}{ccc}
1 & 0 &0\\
2 & 1& 0\\
4 & 3 &1\\
\end{array}
\right)
M_9=\left(
\begin{array}{ccc}
1 & 0 &1\\
3 &1  &0\\
1 & 0 &0\\
\end{array}\right).
$$ 
By Theorem~\ref{goppa}, the set $\MBF M=\{ M_1, \ldots, M_9\}\in \, \mathcal{U}(9,\vec{3},3,5,1)$. We note that $\MBF M$ is
 an example of a genus $1$ UDMG which is not a UDM.
This follows from setting $\lambda_i=0$ for $i=1,2,3,4,5,9$ and $\lambda_6= \lambda_7= \lambda_8=1$, so $\displaystyle \sum_{i=1}^9 \lambda_i =3$, 
and seeing that the resulting matrix $$\left(
\begin{array}{ccc}
1 & 1 &1\\
3 & 2& 2\\
1 & 4 &4\\
\end{array}
\right)
$$ is not of full rank.  
Also note that concatenating the first column of  each $M_i$, $1\leq i\leq 9$, gives
$$
M=\left(
\begin{array}{ccccccccc}
1 & 1 &1 &1 & 1 &1 &1 &1&1\\
0 & 4& 3 & 0 & 4 & 3 & 2 & 2& 3\\
0 & 4 &4 & 3 & 0 & 1 & 4 & 4&1\\
\end{array}
\right)
$$ which is a generating matrix for the Goppa code associated with $(X,\mathbf{P}, D)$ by Theorem \ref{gen}.
\end{example}

It is seemingly non-trivial to check that $\MBF M$ is a UDMG of genus $1$ without using Theorem \ref{goppa}.

\begin{remark}\label{nothing}
It is clear that from our construction of a Goppa UDMG $\MBF M$ associated to a curve $X$ of genus $g$ over $\bF_q$
that its size $L$ is bounded by $\#(X(\bF_q))$.
% [As in the case of Goppa codes, one can do one matrix better if one
%replaces $(K+g-1)P_0$ in the construction by a linearly equivalent divisor $D$ over $\bF_q$ which has no points in its support in $C(\bF_q)$ \cite{koming}.
%Or equivalently, we add the change of basis matrix to a decreasing pole basis at $P_0$.]
% Now we consider the problem of bounding the size of a Goppa UDM with parameters $(\mathcal{C},\mathbf{P},(K+g-1)R).$
%We know for each point on our curve $\mathcal{C}$  in $\mathbb{P}^2(\mathbb{F}_q)$ we can find an increasing zero order basis. Thus for each point on $\mathcal{C}$ in $%\mathbb{P}^2(\mathbb{F}_q)$  we can get a change of basis matrix with our fixed point.
Bounds on the number of points on a curve over a finite field are given by the famous Hasse-Weil-Serre Theorem:
\end{remark}

 \begin{thm}
\label{HW}
Let $X$ be a curve of genus $g$ defined over $\mathbb{F}_q$. Then
$$q+1-g\lfloor 2\sqrt{q} \rfloor \leq \# (X(\mathbb{F}_q)) \leq q+1+g\lfloor 2\sqrt{q}\rfloor.$$
\end{thm}

%Thus we see that for $\mathcal{A}\in \mathcal{U}(L,\MBF N,K,q,g)$ such that $\mathcal{A}$ is a Goppa UDM with parameters $(\mathcal{C},\mathbf{P},(K+g-1)R)$ we have $\#
%\mathcal{A}=L\leq \mathcal{C}(\mathbb{F}_q)$ which must satisfy \ref{HW}.
These bounds are not always sharp, so in order to build  Goppa UDMG of fixed genus $g$ and maximal size, we want to find curves of  genus $g$
over a given finite field that have the maximal known number of rational points. The problem of finding such curves is very well-studied and is a continual area of active research. The role that the work of Tsfasman, Vladut, and Zink on this problem played in the construction of Goppa codes with parameters that beat the Gilbert-Varshamov bound for linear vector codes over finite fields is described in \cite{JW}. For the latest on which curves of which genus
over which finite fields are known to have the most rational points, see the website  \cite{van}.
% Moreover, in \cite{Van} for $q=2,4,8,16,32,64,128,3,9,27,81$ and $1\leq g\leq 50$. These curves give lower bounds for those $L$ for which $\mathcal{U}(L,K,K,q,g)$ is non-%empty [what is the
%right range for $K$ here?].\\
%\indent Note, we have a fundamental trade off between the genus $g$ and the size of the UDM, $L$. Now pertaining to the original engineering problem this addresses, we %see that we can send a vector over a larger number of parallel channels with unique decoding if we have more symbols successfully received. In fact, we can theoretically send %a vector over $q+1$ more parallel erasure channels with each extra symbol received and still have unique decoding. One major unanswered question is, does there exist a %general construction that has a higher achievable lower bound than the Goppa construction? Moreover, are the upper bounds from \ref{bounddelta} or \ref{bound2} achievable %outside of the few special parameters that were given above?
%
%

\section{Upper bounds on the size of UDMGs}

Our first bound comes from our work in section 2. Combining Lemma~\ref{linearbound} with Proposition~\ref{linear}  gives:

\begin{thm}
\label{1dimdelta}

For non-degenerate $\MBF M\in \mathcal{U}(L,\MBF N,K,q,g)$ we have 
$$\#(\MBF M)=L\leq K-2+(g+1)(q+1).$$
\end{thm}
%\begin{proof}
%Let $\mathcal{A}\in \mathcal{U}(L,\{1\},K,q,g)$ be non-degenerate. Non-degeneracy implies that $L\geq K+g.$ Let $G^t$ be the matrix formed by concatenating all of the %vectors in $\MC{A}.$ From the definition of UDMG we see that $G$ is a generating matrix for an $[L,K,d]$ $\MB{F}_q$-linear code, $C$, of Singleton defect $g.$ By Lemma %\ref{linearbound} we have that $L\leq K-2+(q+1)(g+1).$
%\end{proof}
This is only a good bound when $\MBF N=\vec{1}$: we will  now use it to get a better bound for many choices of parameters, by taking the quotient of a UDMG by an appropriate
proper subUDMG to reduce to the case where $\MBF N=\vec{1}$.

\begin{defn} Let $\MBF M$ be a nondegenerate UDMG in $\MC{U}(L,\MBF N,K,q,g).$ Suppose %$N\geq K+g$ and 
that each $N_i\geq 2,$  $1\leq i\leq L$. Then we break such UDMGs into two classes.
If $\sum_{i=1}^L (N_i-1)\geq K-2$, we will say $\MBF M$  is of \textit{Class }$1$. 
%If $\sum_{i=1}^L N_i-1<K-2$, 
Otherwise, we will say $\MBF M$ is of \textit{Class }$2$.
% Note this implies that $w\leq \max\{0,L-2-g\}$. 
\end{defn}

%\begin{defn} With notation as above let $L\geq 1, K\geq N \geq 2,$ such that $LN\geq K.$
%% Then either $r(N-1) \geq K-2$ and $w=0$ or $r(N-1)<K-2$ and $r(N-1)+w=K-2.$ \\
%Fix a $d$ such that $0\leq d\leq K-2.$ Let $B\subset \mathbb{F}_q^K$ such that $\dim B=K-2-d$.
%  Fix an isomorphism $\rho: \mathbb{F}_q^K/B \rightarrow \mathbb{F}_q^{d+2}.$
%\\
%\indent  In the \textit{Case }$1$, $L(N-1)\geq K-2$,
% we have $K-2=P(N-1)+J$ where $0 \leq J \leq N-2$ and $0\leq P \leq L$ by our hypothesis and the division algorithm.
%Let $\mathbf{M}=(M_1,\ldots, M_{L})\in (\mathcal{M}_{K\times N})^{L}.$
%   We define a map $$\hat{\phi}_B:(\mathcal{M}_{K\times N})^{L}\rightarrow (\mathbb{F}_q^{d+2})^{L},$$ 
%\begin{displaymath}
%\hat{\phi}_B(\MBF{M})_i=\left\{ \begin{array}{cc}
%\rho \circ \phi_B((\MBF{M}_i)_N) & 1\leq i \leq P\\
%\rho \circ \phi_B((\MBF{M}_i)_{J+1})& i=P+1 \\
%\rho \circ \phi_B((\MBF{M}_i)_1) & P\leq i \leq L.
%\end{array} \right.
%\end{displaymath}
%\indent In \textit{Case }2, $K-2>L(N-1)$, suppose we have $\MBF{M}=(M_1,\ldots, M_{L-w})\in (\mathcal{M}_{K\times N})^{L-w}$ with $w$ as above. Then we define a map $$\hat{\phi}_B:(\mathcal{M}_{K\times N})^{L-w} \rightarrow (\mathbb{F}_q^{d+2})^{L-w},$$
%\begin{displaymath}
%\hat{\phi}_B(\MBF{M})_i= \left. \begin{array}{cc}
%\rho \circ \phi_B((\MBF{M}_i)_N) & 1\leq i \leq L-w.
%\end{array} \right.
%\end{displaymath}
%\end{defn}

\begin{lem}
\label{quot}
Suppose we have a non-degenerate $\MBF M\in \mathcal{U}(L,\MBF N,K,q,g)$ with each  $N_i\geq 2, $ and that $\MBF M$ is of  \textit{Class }$1$.% , $\sum_{i=1}^L N_i-1\geq K-2$. 
Then there is an integer $d$ with $0\leq d \leq \min( g, K-2)$ such that there exists a corresponding $\tilde{\MBF M} \in \mathcal{U}(L,\vec{1},d+2,q,g-d).$
\end{lem}

\begin{proof} Let $\MBF M$ be as in the statement of the Lemma and  $V(\MBF M)$ be its vector space realization. Let $\nu=(\nu_1,\ldots,\nu_L)$ be such that $\nu_i\leq N_i-1$ and $\sum_{i=1}^L\nu_i= K-2$. Let $\MBF S$ be the proper subUDVSG of $\MBF M$ gotten by truncating the $i^{th}$ chain of $\MBF M$ to a chain of length $\nu_i$.
% \oplus_{i=1}^L V(M_i)_{\nu_i}$.
 % Now let $A$ be the matrix formed by concatenating the first $N-1$ columns of each $M_i\in\mathcal{A}.$
   %Then let $B$ be the space spanned by the first $K-2$ columns of $A$. 
%   So $\rank{B}=K-2-d$ with $0\leq d \leq \min(K-2, g)$ since the genus of $\mathcal{A}$ is $g$. 
%Note that $\tilde{\mathcal{A}}=\hat{\phi}_B(\mathcal{A})$ is a collection of $L$ vectors in $\mathbb{F}_q^{d+2}.$ It should be noted that some of the $L$ vectors could be the $0$ vector.
% Moreover, the span of any $g+2$ vectors in this set must be a $d+2$ dimensional space since $\mathcal{A}\in\mathcal{U}(L,N,K,q,g).$ Thus $\tilde{\mathcal{A}}$ is a UDM of genus $(g+2)-(d+2)=g-d.$ Then we have $\tilde{\mathcal{A}} \in \mathcal{U}(L,1,d+2,q,g-d)$.
By Theorem \ref{quotientUDMG} (which applies with $r=2$), we have that $V(\MBF M)/\MBF S$ is a UDVSG with parameters $(L,N-\nu,d+2,q,g-d)$ 
for some $d$ with  $0\leq d \leq \min(K-2, g)$. Now take $\tilde{\MBF M}$ be the UDMG corresponding to the truncation of $\MBF M/\MBF S$ in which every chain has been truncated to its first subspace. Then $\tilde{\MBF M}\in \MC{U}(L,\vec{1},d+2,q,g-d).$
\end{proof}

\begin{thm} 
\label{bounddelta}
Let $\MBF M$ be a non-degenerate UDMG in  $\mathcal{U}(L,\MBF N,K,q,g)$ with each $N_i\geq 2$. Let $\gamma_{\MBF M}=\min_{i=1}^LN_i$.
%we have the following bound on $\#\mathcal{A}=L$. Again we have two cases.
If $\MBF M$ is of  \textit{Class }$1$, we have 
$$L\leq (g+1)(q+1).$$
Otherwise we have 
$$g+3\leq L \leq  \frac{K-2}{\gamma_{\MBF M}-1}.$$
\end{thm}

\begin{proof}
\indent Suppose $\MBF M$ is of  \textit{Class }$1$, so $N-L=\sum_{i=1}^L (N_i-1)\geq K-2$. %Fix an $\mathcal{A}\in\mathcal{U}(L,\MBF N,K,q,g).$
 By Lemma \ref{quot} we get a corresponding  $\tilde{\MBF M}\in \mathcal{U}(L,\vec{1},d+2,q,g-d)$ for some $0\leq d\leq \min{(g,K-2)}$.
  If $\tilde{\MBF M}$ is degenerate then it must be because $L\cdot 1< (d+2)+(g-d)=g+2\leq (g+1)(q+1).$ 
%  \ref{1dim} and
  If $\tilde{\MBF M}$ is nondegenerate, then by Theorem
   \ref{1dimdelta} we get  $L\leq d+(g-d+1)(q+1)\leq (g+1)(q+1)$. 
   So in either case the result follows.\\
\indent Now suppose  $\MBF M$ is of  \textit{Class }$2$, so $N-L<K-2$. Since $N\geq \gamma_{\MBF M} L$ we get $L<\frac{K-2}{\gamma_{\MBF M-1}}.$
%So, we have $N-L+w=K-2$. Note, $0<w\leq L-2-g.$ 
Finally, since $\MBF M$ is non-degenerate, we have $K+g\leq N<K+L-2$, so subtracting $K$ yields $g+2<L.$
%[Doesn't this give $g+3\leq L$?]  
% By \ref{quotw} we get a non-degenerate $\tilde{\mathcal{A}}\in \mathcal{U}(L-w,1,d+2,q,g-d)$. By 
%% \ref{1dim} and 
% \ref{1dimdelta} we get  $L-w\leq d+(g-d+1)(q+1)\leq (g+1)(q+1)$. Recall, $w=K-2-N+L$ and so we have $N+2-K\leq (g+1)(q+1).$ Futhermore, we have $g+2\leq L \leq \frac{K-2}{\gamma_{\MC{A}}-1}.$
 % So the above implies $N+2-K\leq (g+1)(q+1).$ By hypothesis we have that $L > N+2-K.$ Thus the result follows.\\
\end{proof}

\begin{remark}
%Note in \ref{bounddelta} we have $\frac{K-2}{N-1}\leq \frac{K-2+(g+1)(q+1)}{N}$ unless $(q+1)(g+1)< \frac{K-2}{N-1}.$
%\end{remark}
1) Theorem  \ref{bounddelta} agrees with the bound in Lemma 9 of \cite{VG}, for $\MBF M\in\mathcal{U}(L,\vec{\eta},K,q,0)$ in the region $\eta\leq K \leq 2\eta$. 
Moreover, Theorem \ref{bounddelta} is tighter than the bound in Lemma 10 of \cite{VG} when $K=2\eta+1$ and provides a bound on $L$ for all $\eta\leq K$.
Of course our bound is of a slightly different nature since we assume $\eta \geq 2$ throughout and their bound also includes the $\eta=1$ case. \\

2) If we have an $\MBF M \in \mathcal{U}(L,\vec{2},2,q,0)$ then we can create an $\hat{\MBF M}\in\mathcal{U}(L,\vec{2},2,q,g)$ by  taking $g+1$ copies of each matrix in $\MBF M$. They show in \cite{VG} the existence of an $\MBF M\in\mathcal{U}(q+1,\vec{2},2,q,0)$ and so we see that the bound $L\leq (g+1)(q+1)$ is sharp for UDMG in 
% $\tilde{\mathcal{A}}\in
 $\mathcal{U}(L,\vec{2},2,q,g)$. 
 %\begin{example}
%\label{27UDM}
%So far as we know, up until now there were no $K \times N$ UDMGs of genus $0$ constructed such that $K\geq 2N+3$. Consider the transposes of the following matrices with %entries in $\mathbb{F}_2$.
%$$M_1=\left(
%\begin{array}{ccccccc}
%1 & 0 & 0 & 0 & 0 & 0 & 0\\
%0 &1 & 0 & 0 & 0 & 0 & 0\\
%\end{array}
%\right)
%$$
%$$
%M_2=\left(
%\begin{array}{ccccccc}
%0 &0 &1 & 0 & 0 & 0 & 0\\
%0 &0 & 0 & 1 & 0 & 0 & 0\\
%\end{array}
%\right)
%$$
%$$M_3=\left(
%\begin{array}{ccccccc}
%0 & 0 & 0 & 0 & 1 & 0 & 0\\
%0 & 0 & 0 & 0 & 0 & 1 & 0\\
%\end{array}
%\right)
%$$
%$$
%M_4=\left(
%\begin{array}{ccccccc}
%0 &0 &0 & 0 & 0 & 0 & 1\\
%1 &1 & 1 & 1 & 1 & 1 & 1\\
%\end{array}
%\right)
%$$
%Note the bound given in  Theorem~\ref{bounddelta}
 %is $L \leq 4.$
%Thus the example is tight on the bound. Furthermore, this example can be generalized to any case of the form $K\times N_i$ with $N_i=2$ and $K=2M+1$ for $M\geq 3$. Therefore %we know that the bound is tight on UDMGs with parameters $(M+1,\{2\},2M+1,2,0).$
%\end{example}

3) Theorem \ref{bounddelta}  is not sharp for all classes of UDMGs. 
We now present another bound on $L$ for certain UDMGs and give an example where this new bound is sharper then the bound in Theorem \ref{bounddelta}.
%we can also get a bound for $L$ by considering a corresponding a $(g+1)$-regular UDMG with $K-2=g$ as well. And in doing so, we sometimes get a better upper bound on $L$.\\
\end{remark}
%We now consider a base case, and will provide bounds on $L$ for $\MC{U}(L,\{g+1\},g+2,q,g).$

\begin{lem}
\label{deltabound}
For nondegenerate $\mathcal{A}\in\mathcal{U}(L,\MBF{N},K,q,g)$ with $N_i\geq K-1$ for each $1\leq i\leq L$,  we have the bound $$\binom{K-2+L}{K-1} \leq \binom{K+g-1}{K-1}\frac{q^{K}-1}{q-1}.$$
\end{lem}

\begin{proof}
Let $\MBF M=\{M_1,...,M_L\}$ be in  $\mathcal{U}(L,\MBF{N},K,q,g)$ with $N_i\geq K-1$ for all $i$.  Let $\mathfrak{P}$ be the set of all partitions of $K-1$ into $L$ non-negative integers. %We will let $\mathfrak{C}=\{\lambda_i\}$.% recall the definition of %allowable columns is given in definition \ref{allowable}.
%So each $X_i\in\mathfrak{C}$ corresponds to a different ordered partition of $K-1$ into  $p$ non-negative integers. 
%Thus $\# \mathfrak{C}$ is the number of distinct ordered partitions of $K-1$ into $L$ non-negative integers. 
%This problem is equivalent to counting the number of monomials is $L$ variables of degree $K-1$. The solution is
Then it is well-known that  
\begin{equation}\label{partitions}
\#(\mathfrak{P})=\binom{K-2+L}{K-1}.
\end{equation}
   For each partition $\lambda=(\lambda_1,...,\lambda_L) \in \mathfrak{P},$ let $\displaystyle \Xi(\lambda)$ be the set of columns formed from the first ${\lambda}_j$ columns of $M_j, \ 1\leq j\leq L.$ 

%For details on the above combinatorics solution see \cite{stan}.\\
Since the codimension 1 subspaces of $\bF_q^K$ are in one-to-one correspondence with the points in $\MBF P^{K-1}(\bF_q)$,
 there are $\frac{q^K-1}{q-1}$ subspaces of dimension $K-1$ in $\mathbb{F}_q^{K}.$ Order them arbitrarily
 and let the $j^{th}$ subspace be denoted $S_j, \, 1\leq j\leq \frac{q^{K}-1}{q-1}.$ We define a map $\Upsilon:\mathfrak{P}\rightarrow \mathbb{F}_2^{\frac{q^{K}-1}{q-1}}$ where 
$$\Upsilon(\lambda)_j:=\left\{ \begin{array}{cc} 
0 & \text{if span }sp(\Xi(\lambda))\not\subseteq S_j.\\
1 & \text{if span }sp(\Xi(\lambda)) \subseteq S_j.
\end{array}\right.
$$
%Note that $\Upsilon$ is an indicator map telling us which subspaces of dimension $g+1$ the span of the columns of each $X$ lie in. 
\indent For $1\leq j\leq (q^K-1)/(q-1)$, let $T_j$ be the set of $\lambda\in \mathfrak{P}$ such that $\Upsilon(\lambda)_j=1$, and let $y_j$ be the cardinality of $T_j$.
If the union $U_j$ of $\Xi(\lambda)$ for all $\lambda\in T_j$ contained $K+g$ columns, then since $\MBF M \in \mathcal{U}(L,\MBF{N},K,q,g)$, $sp(U_j)$ would be $K$-dimensional,
which is impossible since $sp(U_j)\subseteq S_j$. Hence $U_j$ contains at most $K+g-1$ columns, so there are at most $\binom{K+g-1}{K-1}$ such $\Xi(\lambda)\in T_j$
and $y_j\leq \binom{K+g-1}{K-1}$.

%Consider $\displaystyle \sum_{\lambda \in \mathfrak{C}} \Upsilon(\lambda)$. Call the $j^{th}$ entry of $\displaystyle \sum_{\lambda \in \mathfrak{C}} \Upsilon(\lambda)$, $y_j$.
% and $y_k> \binom{K+g-1}{K-1}.$ Then there exist $\{\lambda_{j}\}_{j=1}^{y_k}\subset \mathfrak{C}$ such that the span$(X_{1}\ldots X_{{y_k}})\subseteq S_{k}$. 
% If $y_j> \binom{K+g-1}{K-1}$, \ref{allow} tells us that we can pick a set of allowable columns from $\{X_{i}\}_{i=1}^{y_j}$ of size $K+g$ and their span will lie in $S_{j}$ %which contradicts $\mathcal{A}\in \mathcal{U}(L,\MBF{N},K,q,g).$
%Thus we must have $$ y_j\leq \binom{K+g-1}{K-1}.$$ This means for each distinct $S_j$ there can be at most $\binom{K+g-1}{K-1}$ distinct $\lambda \in \mathfrak{C}$ %such that $\Upsilon(\lambda)_j\neq 0.$ Since $\Upsilon(\lambda)$ has at least one non-zero entry we have
Note that $\Upsilon(\lambda)$ is guaranteed to have at least one non-zero entry since $\Xi(\lambda)$ is a set of $K-1$ columns.
 Putting this together we have:
 $$\#( \mathfrak{P}) \leq \sum_{\lambda\in\mathfrak{P}}\sum_{j=1}^{\frac{q^{K}-1}{q-1}} \Upsilon(\lambda)_j= $$ $$  \sum_{j=1}^{\frac{q^{K}-1}{q-1}}\sum_{\lambda\in\mathfrak{P}} \Upsilon(\lambda)_j\leq  \sum_{j=1}^ {\frac{q^{K}-1}{q-1}} \binom{K+g-1}{K-1}\leq  \binom{K+g-1}{K-1}\frac{q^{K}-1}{q-1}.$$ \\
\indent Thus by the formula for $\#( \mathfrak{P})$ in (\ref{partitions}) we have
$$\binom{K-2+L}{K-1} \leq \binom{K+g-1}{K-1}\frac{q^{K}-1}{q-1}.$$
\end{proof}

Note that Lemma \ref{deltabound} gives an upper bound on $L$ because  the left hand side is a polynomial in $L$, which is increasing as a function of the positive integers, and  the right hand side is a number depending only on the other parameters of the UDMG.

\begin{example} Consider a non-degenerate $\MBF M\in \mathcal{U}(L,\vec{\eta},4,2,2)$ with $\eta\geq 4,$
%An example where \ref{bound2} is a tighter upper bound on $r$ than \ref{bounddelta} is for $\mathcal{A}\in \mathcal{U}(r,N,N,2,2)$ when $r\geq 1$ and $N\geq 4.$ Note 
so we have $L(\eta-1)=N-L\geq K-2$ and $\MBF M$ is of class 1.
Then from Theorem \ref{bounddelta} we get $L\leq 9.$ 
We can apply Lemma \ref{deltabound} to get the upper bound $\binom{2+L}{3}\leq 150,$ which  implies $L\leq 8.$  Thus there are cases when Lemma
\ref{deltabound} gives a sharper upper bound than Theorem \ref{bounddelta}.
\end{example}

%\begin{remark} 
%We presume that the general technique of bounding $L$ by using quotient spaces can be applied to reduce to any UDM with smaller parameters, namely for $\MBF N'\leq \MBF %N$ but we only present the above because the combinatorics were nice.
 %\end{remark}

\bibliographystyle{amsplain}
%\bibliography{/Users/stevelimburg/Desktop/TexStuff/Bib}
\bibliography{Bib1}

\end{document}